\documentclass[11pt]{article}
\usepackage{amsmath, amsthm, amssymb, amsfonts}

\usepackage{graphicx, subcaption}
\usepackage{cite}
\usepackage{makecell}
\usepackage[letterpaper, left=1in,right=1in,top=1in,bottom=1in]{geometry}

\usepackage{tikz, bbding, tikz-3dplot}
\usetikzlibrary{patterns}
\usetikzlibrary{matrix}
\usetikzlibrary{tikzmark}
\usetikzlibrary{positioning}
\usetikzlibrary{fit}
\usetikzlibrary{shadows.blur}
\usetikzlibrary{shapes.symbols}
\usetikzlibrary{shapes.geometric}
\usetikzlibrary{calc, decorations.pathreplacing}
\usepackage{makecell}

\newtheorem{lemma}{Lemma}
\newtheorem{theorem}{Theorem}
\newtheorem{definition}{Definition}

\usepackage{cancel}
\usepackage{enumitem}

\allowdisplaybreaks

\def \chg {\mathcal{G}}
\def \groupW {\mathcal{W}}
\def \nologP {no( \log \bar{P})}
\def \bP {\bar{P}}
\def \bA {\boldsymbol{A}}
\def \lmu {\underline{\mu}}
\def \umu {\overline{\mu}}
\def \bAmu {(\bA)^{\lmu}}
\def \bBmu {(\bB)^{\umu}}
\def \bB {\boldsymbol{B}}
\def \bX {\boldsymbol{X}}
\def \bY {\boldsymbol{Y}}
\def \bT {\boldsymbol{T}}
\def \bU {\boldsymbol{U}}
\def \bV {\boldsymbol{V}}
\def \bbY {\overline{\boldsymbol{Y}} }

\def \Hg {H_\chg}
\def \Ig {I_\chg}

\newcommand{\lrfloor}[1]{\left\lfloor {#1} \right\rfloor}
\newcommand{\lrbr}[1]{\left\{ #1 \right\}}
\newcommand{\lrpar}[1]{\left( #1 \right)}
\newcommand{\lrbkt}[1]{\left[ #1 \right]}
\newcommand{\lrabs}[1]{\left| #1 \right|}

\newcommand{\sP}[1]{\sqrt{P^{#1}}}

\title{Secure GDoF of the $Z$-channel with Finite Precision CSIT: \\How Robust are Structured Codes?}
\author{Yao-Chia Chan and Syed A. Jafar\\
	{\small Center for Pervasive Communications and Computing (CPCC)}\\
	{\small University of California Irvine, Irvine, CA 92697}\\
	{\small \it Email: \{yaochic, syed\}@uci.edu}
}
\date{}

\begin{document}
	\maketitle
	
	\begin{abstract}
		Under the assumption of perfect channel state information at the transmitters (CSIT),  it is known that structured codes offer significant  advantages for secure communication in an interference network, e.g.,  structured jamming signals based on lattice codes  may allow a receiver to  decode the sum of the jamming signal and the signal being jammed, even though they cannot be separately resolved due to secrecy constraints,  subtract the aggregate jammed signal, and then proceed to decode desired codewords at lower power levels. To what extent are such benefits of structured codes fundamentally limited by uncertainty in CSIT? To answer this question, we explore what is perhaps the simplest setting where the question presents itself --- a $Z$ interference channel with secure communication. Using sum-set inequalities based on Aligned Images bounds we prove that the GDoF benefits of structured codes are lost completely under finite precision CSIT. The secure GDoF region of the $Z$ interference channel is obtained as a byproduct of the analysis.
	\end{abstract}
	{\let\thefootnote\relax\footnote{{This work is supported  by NSF grants CCF-1617504 and CNS-1731384,  ARO  grant W911NF-19-1-0344 and ONR grant N00014-18-1-2057. It was presented in part at the 2020 IEEE International Symposium on Information Theory \cite{Chan_Jafar_ZIC}. }}\addtocounter{footnote}{-1}}
	\newpage	
	
	\section{Introduction} \label{sec:intro}
The capacity  of wireless networks, as evident from recent Degrees of Freedom (DoF) \cite{Zheng_Tse_Diversity} and Generalized Degrees of Freedom (GDoF)\cite{Etkin_Tse_Wang} studies, depends rather strongly on the underlying assumptions about the availability of channel state information at the transmitter(s) (CSIT). Zero forcing\cite{Caire_Shamai_ZF, Weingarten_Shamai_GBC}, interference alignment\cite{Jafar_Shamai, Cadambe_Jafar_int, Jafar_FnT, Motahari_Gharan_Khandani_real} --- structured codes\cite{Bresler_Parekh_Tse, Nazer_Gastpar_Compute} in general --- are powerful ideas;  nevertheless their benefits can quickly disappear under even moderate amounts of channel uncertainty. Robustness is  paramount, and it is enforced in GDoF studies by limiting CSIT to finite precision \cite{Lapidoth_Shamai_Wigger_BC, Arash_Jafar}. This leads naturally to a  crucial question:  {\it how robust are structured codes?} Specifically, to what extent does finite precision CSIT fundamentally limit the benefits of structured coding schemes? The question is important from both practical and theoretical perspectives. The emphasis on finite precision CSIT brings  theory closer to practice, which is a worthy goal in itself. In addition, even if we set practical concerns aside, there is another motivation for the emphasis on robustness --- if the benefits of structured codes are indeed lost under finite precision CSIT, then perhaps this removes some of the obstacles that have made progress difficult in network information theory, and thus opens the door to a comprehensive and robust network information theory of wireless networks, based on optimality of random codes that are much better understood.

Under perfect CSIT the challenge in GDoF studies is the crafting of powerful  achievable schemes. Finite precision CSIT shifts the challenge to  \emph{outer bounds}. Indeed,  optimal schemes under finite precision CSIT tend to be  classical random coding schemes that are well understood. What is difficult is to prove that these schemes are \emph{optimal}, e.g., that alignment is not possible, that \emph{nothing} more powerful exists (in the GDoF sense) under finite precision CSIT. Another motivation for the focus on GDoF outer bounds is that unlike  inner bounds that are inherently cumbersome as they depend on numerous design choices, e.g., number of layers of rate-splitting for each user,  the rates and power levels assigned to each layer, and various choices of spatial and temporal beamforming, GDoF outer bounds tend to be \emph{much} more compact, depending only on the channel parameters.

Accounting for  arbitrary structure is essential because, unlike random noise, interference can be arbitrarily structured. It is the structure of the codes that decides how the signals align with each other, how many signal dimensions they occupy together, whether they add constructively or destructively, whether they can be collectively or individually decoded \cite{Etkin_Ordentlich, Jafar_Vishwanath_GDOF,  He_Yener_M21, He_Yener_2usrGaussian, Xie_Ulukus, Xie_Ulukus_K, Mukherjee_Xie_Ulukus, Mukherjee_Ulukus_MIMOWiretap,   Banawan_Ulukus_misbehavior, Geng_Chen_WiretapHelper, Chen_Li_Helper, Li_Chen_CommonRandom, Chen_Jam}. Accounting for structure, even from the coarse GDoF perspective, turns out to be difficult, perhaps because structured codes are inherently combinatorial objects. This is especially the case for \emph{robust} GDoF studies (e.g., with CSIT limited to finite precision), where it is increasingly evident that classical information theoretic tools are lacking. With the exception of  `Aligned Images (AI)' bounds \cite{Arash_Jafar}, there are no alternatives, to our knowledge, that have been found to be capable of bounding the benefits of structure under non-trivial channel uncertainty.  For example, aside from the combinatorial approach of AI bounds, there still is no other argument to \emph{prove} that the $K$ user interference channel has \emph{any} less than a total of $K/2$ DoF under finite precision CSIT. Note that Aligned Images bounds can  prove something \emph{much} stronger --- that it has only a total of $1$ DoF\cite{Arash_Jafar}. In fact even if all the transmitters cooperate fully the resulting $K$ user MISO BC still has only $1$ DoF (thus resolving a conjecture by Lapidoth, Shamai and Wigger\cite{Lapidoth_Shamai_Wigger_BC}). AI bounds have been similarly essential to robust GDoF characterizations of various interference and broadcast settings, such as the symmetric $K$ user IC \cite{Arash_Jafar_IC}, the $2$-user MIMO IC with arbitrary levels of CSIT \cite{Arash_Jafar_mimoSymIC}, the $3$ user MISO BC \cite{Arash_Jafar_SLS}, and the $2$-user MIMO BC with arbitrary levels of CSIT,\cite{Arash_Bofeng_Jafar_BC, Arash_Jafar_MIMOBC}. Robust GDoF characterizations have also been found using AI bounds for various intermediate levels of transmitter cooperation  in \cite{Arash_Jafar_cooperation, Chan_Wang_Jafar_Extremal,  Wang_Jafar_LimitedCoop}. 

 
Aligned Images bounds are so called because they are based on counting the expected number of codewords that can cast  ‘aligned images’ at one receiver while casting resolvable images at another. Because of their essentially combinatorial character,  derivations of AI bounds can be somewhat tedious. Yet, the lack of alternatives thus far makes these bounds indispensable to the goal of developing a robust understanding of the capacity limits of wireless networks. In order to make further progress in this  direction, it is important to explore and expand the scope of AI bounds. Notably, the class of AI bounds was recently  expanded significantly into a broad class of sum-set inequalities in \cite{Arash_Jafar_sumset}. Exploring applications of these increasingly sophisticated sumset inequalities is another motivation for our work in this paper.

With the aid of sumset inequalities we wish to explore the robustness of structured codes for secure communication \cite{Karpuk_Chorti_PLN,Xie_Ulukus, Xie_Ulukus_K, Mukherjee_Xie_Ulukus, Banawan_Ulukus_TimeVarying, He_Yener_M21, Papadimitratos_Scalable, Skoglund_mixedCSIT, Geng_Tandon_Jafar_DetSecureIC, Chen_Li_Helper, Chen_Jam}. In particular, one powerful idea that is made possible by structured codes is the aggregate decoding and cancellation\footnote{ `Aggregate decoding and cancellation' is used loosely here to refer to any means by which the interference  from  jammed signals at higher power levels to the desired signals at lower power levels can be mitigated. The focus is  on mitigating the residual interference to lower power levels, and not on the aggregate decoding of higher levels \emph{per se}.} of jammed signals \cite{He_Yener_M21, Papadimitratos_Scalable, Skoglund_mixedCSIT, Geng_Tandon_Jafar_DetSecureIC, Chen_Li_Helper, Chen_Jam}.  Lattice coded  jamming signals are sometimes used to guarantee the secrecy of a message that is itself encoded with a compatible lattice code. A key advantage of structured codes in such settings is that even though neither the jamming noise nor the message is individually decodable, their sum can still be `decoded' and cancelled. Intuitively, this is because the sum of lattice points is still a valid lattice point. The ability to  decode and cancel jammed signals in aggregate is important because it then allows a receiver to successively decode \cite{Cover72} desired signals at lower power levels. However, this ability may not be  robust to channel uncertainty, which is especially a concern for secure communication applications where robustness is paramount. The question is fundamental and therefore broadly relevant, but in order to minimize distractions we  study what is perhaps the simplest scenario where the question presents itself --- a $Z$ interference channel with secrecy constraints \cite{Yates_Trappe_SecureZ, He_Yener_BoundGIC, Bustin_Poor_SecureZ, Mohapatra_Murthy_Lee_SecureZ, Karmakar_Ghosh_SecureFadingZ}. 

\begin{figure*}[t]
	\begin{center}
		\scalebox{0.75}{%
			\begin{tikzpicture}
			\foreach \m in {1,2}
			{
				\coordinate (Mu\m) at (1,2-3*\m);
				\coordinate (Md\m) at (1,1.3-3*\m);
				\coordinate (Nu\m) at (5,2-3*\m){};
				\coordinate (Nd\m) at (5,1.3-3*\m){};
				
			};
			
			\coordinate (Nu10) at (5,-0.3){};
			
			\coordinate (Mu20) at (1,-3.3){};

			\coordinate (Nu20) at (5,-3.3){};

			\draw [rounded corners, thick] (0,-0.6) rectangle (1.5,-2);
			\node at (0.75, -0.6) [above] { $\overline{X}_1$};
			
			\draw [rounded corners, thick] (0,-3) rectangle (1.5,-5);
			\node at (0.75, -5) [below] { $\overline{X}_2$};
			
			\draw [rounded corners, thick] (4.5,0) rectangle (6.9,-2);
			\node at (5.3, 0) [above] { $\overline{Y}_1$};
			
			\draw [rounded corners, thick] (4.5,-3) rectangle (6,-4.3);
			\node at (5.3, -4.3) [below] { $\overline{Y}_2$};

			\node at (Mu1) [left=0.2cm] { $A_1$};
			\node at (Md1) [left=0.2cm] { $A_2$};
			\node at (Mu20) [left=0.2cm] { $B_1$};
			\node at (Mu2) [left=0.2cm] { $B_2$};
			\node at (Md2) [left=0.2cm] { $B_3$};
			\node at (Nu10) [right=0.2cm] { $B_1$};
			\node at (Nu1) [right=0.2cm] { $A_1+B_2$};
			\node at (Nd1) [right=0.2cm] { $A_2+B_3$};
			\node at (Nu20) [right=0.2cm] { $B_1$};
			\node at (Nu2) [right=0.2cm] { $B_2$};
			
			\draw [very thick] (Mu20)--(Nu20);
			\draw [very thick] (Mu1)--(Nu1);
			\draw [very thick] (Md1)--(Nd1);
			\draw [very thick] (Mu2)--(Nu2);
			\draw [very thick, red] (Mu20)--(Nu10);
			\draw [very thick, red] (Mu2)--(Nu1);
			\draw [very thick, red] (Md2)--(Nd1);
			
			\draw [black, thick, fill=blue!20!white] (Mu2)+(-0.1,0) circle (5pt);
			
			\draw [black, thick, fill=red!20!white] (Md1)+(-0.1,0) circle (5pt);
			\draw [black, thick, fill=blue!20!white] (Nu1)+(0.1,0) circle (5pt);
			\draw [black, thick, pattern=crosshatch] (Nu1)+(0.1,0) circle (5pt);
			\draw [black, thick, fill=red!20!white] (Nd1)+(0.1,0) circle (5pt);
			\draw [black, thick, fill=blue!20!white] (Nu2)+(0.1,0) circle (5pt);
			\draw [black,  thick, pattern=crosshatch] (Mu1)+(-0.1,0) circle (5pt);
			\draw [black, thick, fill=white] (Md2)+(-0.1,0) circle (5pt);
			
			\draw [black, thick, fill=white] (Nu10)+(0.1,0) circle (5pt);
			
			\draw [black, thick, fill=white] (Mu20)+(-0.1,0) circle (5pt);

			\draw [black, thick, fill=white] (Nu20)+(0.1,0) circle (5pt);
			
			\end{tikzpicture}
			
			\hspace{0.3in}

			\begin{tikzpicture}[scale=1.2]
			\foreach \m in {1,2}
			{
				\coordinate (M\m) at (1,1.5-2*\m);
				\coordinate (N\m) at (5,1.5-2*\m){};
			};
			
			\draw [very thick] (M1)--(N1) node [pos=0.3, above = 0.1cm] {$\alpha_{11}=1$};
			\draw[very thick] (M2)--(N2) node [pos=0.7, below = 0 cm, text  = black] { $\alpha_{22}=1$};
			
			\draw[very thick, red] (M2)--(N1) node [pos=0.7, below = 0.2cm, text=black] {$\alpha_{12}=3/2$};

			\node[thick, circle, draw=black, fill=white, inner sep = 1.5, left] at  (M1) {$\overline{X}_1$};
			\node[thick, circle, draw=black, fill=white, inner sep = 1.5, right] at (N1){$\overline{Y}_1$};
			
			\node[thick, circle, draw=black, fill=white, inner sep = 1.5, left] at  (M2) {$\overline{X}_2$};
			\node[thick, circle, draw=black, fill=white, inner sep = 1.5, right] at (N2){$\overline{Y}_2$};

			\draw  (-0.75,-0.5) rectangle (-0.25, 0.5);
			\draw  [pattern=crosshatch](-0.75,0) rectangle (-0.25, 0.5) node[label={[xshift=-0.9cm, yshift=-0.7cm]$A_1$}] {};
			\path [thick] (-1,0.5)--(0,0.5) node[above, midway]{$\begin{matrix}\overline{X}_1\in\mathcal{X}_1\\
				\mathcal{X}_1=\{0,1,2,\cdots,\lfloor\sqrt{P^1}\rfloor-1\}\end{matrix}$};
			\draw  [fill=red!20!white](-0.75,-0.5) rectangle (-0.25, 0) node[label={[xshift=-0.6cm, yshift=-0.5cm]$A_2$}, pos=0.5] {?};
			\draw [thick](-1,-0.5)--(0,-0.5) ;
			
			\begin{scope}[shift={(0.3,0)}]
			\draw  (6,-0.5) rectangle (6.5, 0.5);
			\node at (6.6,1.2) [above]{ $\begin{matrix}\overline{Y}_1=\overline{X}_2 \boxplus \overline{X}_2\\ =\lfloor G_{11}\overline{X}_1\rfloor+\lfloor G_{12}\overline{X}_2\rfloor\end{matrix}$};
			\draw  [pattern=crosshatch](6,0) rectangle (6.5, 0.5) node[label={[xshift=-0.87cm, yshift=-0.7cm]$A_1$}]  {};;
			\draw  [fill=red!20!white](6,-0.5) rectangle (6.5, 0) node[label={[xshift=-0.63cm, yshift=-0.57cm]$A_2$}, pos=0.5]{$?$};
			\draw  [fill=white](6.5,0.5) rectangle (7, 1) node[label={[xshift=0.3cm, yshift=-0.7cm]$B_1$}]  {};
			\draw  [fill=blue!20!white](6.5,0) rectangle (7, 0.5) node[label={[xshift=0.3cm, yshift=-0.7cm]$B_2$}]  {};
			\draw  [fill=white](6.5,-0.5) rectangle (7, 0) node[label={[xshift=0.3cm, yshift=-0.7cm]$B_3$}]  {};;
			\draw [thick](5.5,-0.5)--(7.5,-0.5) ;
			\end{scope}
			
			\draw  (-0.75,-2.5) rectangle (-0.25, -1) node[label={[xshift=-0.9cm, yshift=-0.7cm]$B_1$}]  {};
			\draw  [fill=blue!20!white](-0.75,-2) rectangle (-0.25, -1.5) node[label={[xshift=-0.9cm, yshift=-0.7cm]$B_2$}]  {};
			\draw  [fill=white](-0.75,-2.5) rectangle (-0.25, -2) node[label={[xshift=-0.9cm, yshift=-0.7cm]$B_3$}]  {};
			\draw [thick](-1,-2.5)--(0,-2.5)node[below, midway]{$\begin{matrix}\overline{X}_2\in\mathcal{X}_{3/2}\\
				\mathcal{X}_{3/2}=\{0,1,2,\cdots,\lfloor\sqrt{P^{3/2}}\rfloor-1\}\end{matrix}$};
			\node at (6.9,-2.5) [below]{ $\overline{Y}_2=(\overline{X}_2)_{1}^3$};
			\draw  [fill=white](6.5,-2) rectangle (7, -1.5) node[label={[xshift=0.3cm, yshift=-0.7cm]$B_1$}]  {};
			\draw  [fill=blue!20!white](6.5,-2.5) rectangle (7, -2) node[label={[xshift=0.3cm, yshift=-0.7cm]$B_2$}]  {};
			\draw [thick](5.75,-2.5)--(7.75,-2.5);
			
			\end{tikzpicture}
		}
		\caption{\it \small A   toy example. On the left is the ADT deterministic model CSIT which shows that under perfect CSIT the Secure GDoF tuple $(1/2,1/2)$ is achievable (needs lattice alignment between structured codes $B_2$ and $A_1$). On the right is the corresponding channel model under finite precision CSIT, for which we prove in this work that the GDoF tuple $(\delta,1/2)$ is not achievable for any $\delta>0$. This can be seen from Theorem \ref{thm:gdof} by substituting $\beta=3/2, d_2=3/2$ in Case $2$, which yields $d_1\leq 0$.
			Some of the notations are defined in Section \ref{sec:converse}.			
		}\label{fig:ADTAIS}
	\end{center}
\end{figure*}
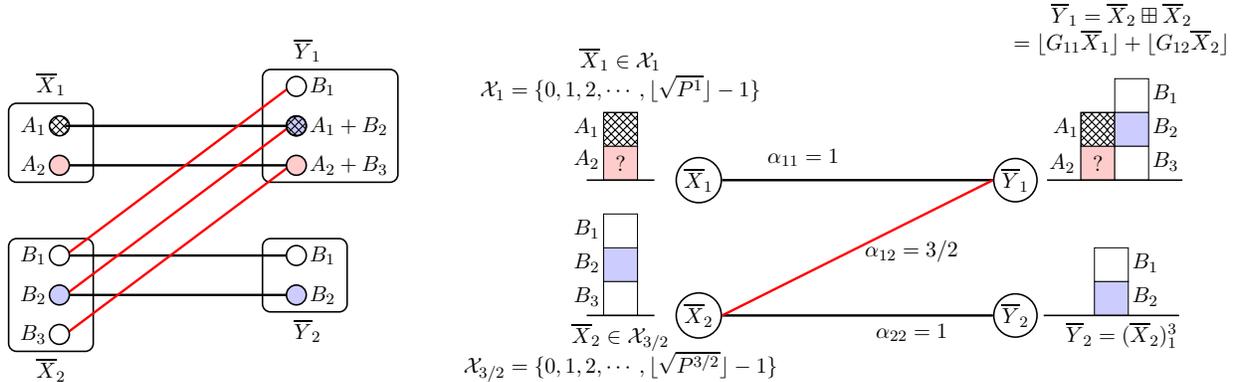

As a motivating example, consider the toy setting of a $Z$ channel illustrated in Figure \ref{fig:ADTAIS} where the two transmitters wish to send independent secret messages to their respective receivers, and only Receiver $1$ experiences interference. The desired links of each user by themselves are capable of carrying $1$ GDoF, while the cross-link has $3/2$ GDoF. Intuitively, if we think of $C_{ij}$  as representing the capacity of the point to point Gaussian channel between Transmitter $j$ and Receiver $i$, then we have $C_{11}:C_{12}:C_{22} = 2:3:2$ for this toy example. Note that the ratios of link capacities correspond to the $\alpha_{ij}$ values in the GDoF model, and that only the relative values of $\alpha_{ij}$ matter\footnote{It follows from the definition of GDoF that if all $\alpha_{ij}$ values are scaled by the same constant then the GDoF value is scaled by that  constant as well.} for the GDoF metric. Throughout this paper we will normalize $\alpha_{22}$ to unity. In the figure\footnote{Intuitively, $\overline{X}_1$, $\overline{X}_2$ are non-negative integers that can be (approximately) expressed in $\lfloor \sqrt{P^{1/2}}\rfloor$-ary symbols as $\overline{X}_1=A_1A_2$ and $\overline{X}_2=B_1B_2B_3$.} we see both the ADT deterministic model \cite{Avestimehr_Diggavi_Tse} (on the left), which implies \emph{perfect} CSIT, as well as the more general deterministic\footnote{The model is not fully \emph{deterministic} in a strict sense, because the channel coefficients are not perfectly known to the transmitters. The nomenclature comes from the fact that the Gaussian noise is removed in this model.} model (on the right)  that allows us to study finite precision CSIT.  Similar to the normalization, $\alpha_{22}=1$, all channel capacities are normalized by the capacity of the channel between Transmitter $2$ and Receiver $2$ in the ADT model. The ADT model shows, intuitively, how it is possible with perfect CSIT to achieve the GDoF tuple $(1/2,1/2)$. Since communication must be secure and the top signal level $B_1$ is fully exposed to the undesired receiver, while the bottom signal level $B_3$ cannot be heard by the desired receiver (below the noise floor) this leaves Transmitter $2$ only $B_2$ to achieve its $1/2$ GDoF. Transmitter $1$ sends a jamming signal $A_1$ to secure $B_2$ from Receiver $1$. The most important aspect of this toy example is the alignment that takes place between $A_1$ and $B_2$, both of which are structured (lattice) codes, so that the sum $A_1+B_2$ also has a lattice structure. This allows Receiver $1$ to `decode' the sum $A_1+B_2$ (without being able to decode $A_1$ or $B_2$ separately, which would violate secrecy), subtract it from the received signal and then decode its desired signal $A_2$ in order to simultaneously achieve $1/2$ GDoF. Now consider the same problem under finite precision CSIT, which poses obstacles for lattice alignment. If lattice alignment is restricted then so is the ability of Receiver $1$ to `decode' the linear combination of signals $A_1$ and $B_2$, which in turn limits the potential for decoding the desired signal $A_2$ that appears at a lower power level. But how strong are these restrictions? Is it still possible to partially mitigate interference from aligned signals at higher power levels to allow decoding of desired signals at lower power levels? Are these restrictions fundamental --- could there be other structured coding schemes,  yet to be discovered, that could overcome such limitations? These are the fundamental questions that motivate this work. What we find, using Aligned Images bounds and sum-set inequalities\cite{Arash_Jafar_sumset}, is that indeed the limitations imposed on structured codes by finite precision CSIT, are both strong and fundamental. In the specific context of this toy example, we {prove}  that the GDoF tuple $(\delta,1/2)$ is not achievable for any $\delta>0$. Thus, the GDoF benefits of lattice alignment, aggregate decoding and cancellation are all lost under finite precision CSIT, underscoring their fragile nature. Moreover, because the bound is information theoretic, no better alternative can exist. Beyond the toy example, the  general proof formalizes the intuition that under finite precision CSIT, lower layers cannot be decoded without decoding higher layers, and higher layers cannot be decoded in aggregate if they cannot be decoded separately. As a byproduct of this analysis, we fully characterize the secure GDoF region of the $Z$ channel under finite precision CSIT.

Since the $Z$-interference channel is a canonical setting that has been extensively studied under a variety of assumptions, let us note that there are three essential distinguishing aspects of our work: 1) robustness, 2) information theoretic optimality in the GDoF sense, and 3) security. It is the combination of these $3$ aspects that makes our setting uniquely challenging and allows us to explore the limitations of aggregate decoding for structured jamming under channel uncertainty. In fact it is arguably the simplest problem that allows us to do so. For example, if we relax any of 	these three constraints then there would be no need for AI bounds. If we relax the robustness constraint by allowing perfect CSIT, then the problem has been studied in \cite{ Yates_Trappe_SecureZ, He_Yener_BoundGIC}, and since channel uncertainty is not a concern, ADT models can be used to construct powerful lattice alignment solutions as shown in Figure \ref{fig:ADTAIS}. If we do not insist on information theoretic optimality then achievable schemes are easily developed, say from \cite{Liu_Maric_Spasojevic_Yates}. If we stop short of GDoF, e.g., only ask for DoF (degrees of freedom) by restricting $\alpha=\beta=1$,  then the problem  becomes trivial because the DoF region is the simplex bounded by $d_1+d_2\leq 1$ even with perfect CSIT, which is also achievable with finite precision CSIT.   If we relax the security constraint, then there is no need for structured codes (e.g., lattice alignment) and the capacity has been characterized within a gap of a constant number of bits in \cite{Zhu_Guo_Z}. Furthermore,  the $2$ user $Z$ interference channel with secrecy constraint is especially appealing because it  has very few channel parameters, which allows us to seek a comprehensive GDoF characterization for the entire parameter space without any assumptions of symmetry, and at the same time the secrecy constraint ensures that the problem is non-trivial and allows room to explore sophisticated applications of the new sumset inequalities \cite{Arash_Jafar_sumset}. Remarkably, despite its simplicity, the $2$-user $Z$-channel is not far from exhausting the scope of known sum-set inequalities. It is noted recently in \cite{Chan_Jafar_3to1} that even if we  introduce just one more user, which changes the $2$-user $Z$ channel into a $3$-to-$1$ interference channel (only Receiver $1$ experiences interference), then the problem of characterizing the secure GDoF region under robust CSIT assumptions may be beyond the reach of known sum-set inequalities.  Finally, let us note that the $Z$-interference channel  has also been explored under  other assumptions that are not so closely related to this work, e.g.,   deterministic encoders \cite{Bustin_Poor_SecureZ}, cooperation between transmitters \cite{Mohapatra_Murthy_Lee_SecureZ}, cooperation between receivers \cite{Fayed_Lai_CoopRx}, binary alphabet  \cite{Karmakar_Ghosh_SecureFadingZ}, and lack of coordination/trust between transmitters \cite{Xie_Ulukus_Game}.

The rest of this paper is organized as follows. 
The system model is presented in the next section. 
The main result, i.e., the secure GDoF region is presented in Section \ref{sec:result}. 
The achievability proof of the main result appears in Section \ref{sec:proof}, and the conserve proof follows in Section \ref{sec:converse} along with a brief review of AI bounds. 
In Section \ref{sec:conclusion} we present the conclusion.

\emph{Notation: } For a positive integer $n$, denote $[n] = \{1,2,\cdots, n\}$.
The set  $\{X(t): t \in [n]\}$ is denoted as $\boldsymbol{X}$. 
For two functions $f(x)$ and $g(x)$, denote $f(x) = o(g(x))$ if $\limsup_{x \rightarrow \infty} f(x)/g(x) = 0$, and $f(x) = O(g(x))$ if $\limsup_{x \rightarrow \infty} f(x)/g(x) = c$ for some constant $c > 0$. 
For random variables $X, Y$ and $Z$, and a set $\chg$, define $\Hg(X|Y) = H(X|Y, \chg)$, and $\Ig(X;Y|Z) = I(X;Y|Z,\chg)$.
All logarithms are to the base 2.
	
\section{System Model}\label{sec:model}
\subsection{The Gaussian $Z$ Interference Channel (ZIC)}
\begin{figure}[t]
\centering
\begin{tikzpicture}[scale=3]
					\def \r {0.07}
					\def \w {1}
					\def \h {0.5}
					
					\draw [fill=white] (0,{-\h*(1-1)}) circle (\r)  node { 1 } node [left=0.1 ] {\footnotesize $W_1 \rightarrow \boldsymbol{X}_1 \rightarrow$};
					\draw [fill=white] (0,{-\h*(2-1)}) circle (\r)  node { 2 } node [left=0.1] {\footnotesize$W_2 \rightarrow \boldsymbol{X}_2 \rightarrow$};
					
					\draw [fill=white] (\w,{-\h*(1-1)}) circle (\r)  node { 1} node [right=0.1] {\footnotesize$\rightarrow \boldsymbol{Y}_1 \rightarrow W_1$ };
					\draw [fill=white] (\w,{-\h*(2-1)}) circle (\r)  node { 2} node [right=0.1] {\footnotesize $\rightarrow \boldsymbol{Y}_2\rightarrow W_2, \xcancel{W_1}$ };
					
					\draw [line width=1] (0+\r,{-\h*(1-1)} ) -- (\w-\r,{-\h*(1-1)}) node [pos = 0.5, above] {\footnotesize $\alpha_{11} = \alpha$};
					\draw [line width=1] (0+\r,{-\h*(2-1)} ) -- (\w-\r,{-\h*(2-1)}) node [pos = 0.5, below] {\footnotesize $\alpha_{22} = 1$};
					
					\draw [line width=1] (0+\r,{-\h*(2-1)} ) -- (\w-\r,{-\h*(1-1)}) node [pos = 0.4, right=0.1] { \footnotesize $\alpha_{12}= \beta$};
					
				\end{tikzpicture}
				\caption{\it \small The Gaussian $Z$ Interference Channel (ZIC).}
				\label{fig:ZIC}
			\end{figure}
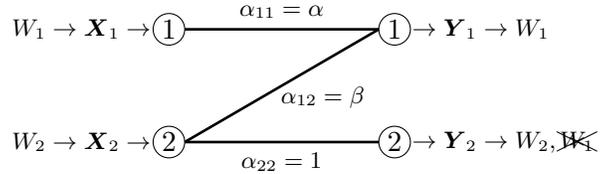
			We consider the two user Gaussian $Z$ Interference Channel depicted in Figure \ref{fig:ZIC}, which consists of two transmitters and two receivers, each equipped with a single antenna. As shown in the figure, the network has a $Z$-topology, so both transmitters are heard by Receiver $1$, while only Transmitter $2$ is heard by Receiver $2$. There are two independent messages $W_1$ and $W_2$, that originate at Transmitter $1$ and Transmitter $2$ and are desired by Receiver $1$ and Receiver $2$, respectively. Message $W_i$ is uniformly distributed over the set $\mathcal{W}_i$. The messages are encoded into codewords $\bX_1, \bX_2$, where $\bX_i =\big( X_i(t)\big)_{t\in[n]}\in\mathbb{R}^{n}$ is a  codeword spanning $n$ channel uses that is sent from Transmitter $i$, and satisfies a unit transmit power constraint, $\frac{1}{n}\sum_{t \in [n]} \mathbb{E}[ |X_i(t)|^2 ] \leq 1$, $i = 1,2.$ The messages are encoded separately and there is no common randomness shared between transmitters; i.e., $\bX_i = f_{i,n}(W_i, \theta_i)$, where $f_{i,n}(.)$, $i = 1, 2$ are encoding functions,  $\theta_i$ is private randomness available only to Transmitter $i$, and $I(\theta_1,W_1; \theta_2,W_2) = 0$. 
			
\subsection{The Gaussian $Z$ Broadcast Channel (ZBC)}		
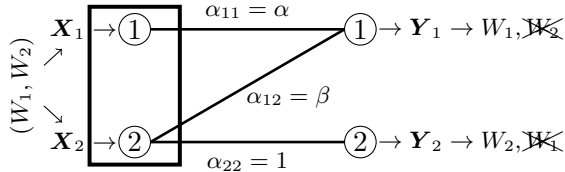
\begin{figure}[t]
				\centering
					\begin{tikzpicture}[scale=3]
						\def \r {0.07}
						\def \w {1}
						\def \h {0.5}
						
						\draw [fill=white] (0,{-\h*(1-1)}) circle (\r)  node { 1 } node (A) [left = 0.1]{\footnotesize $\boldsymbol{X}_1 \rightarrow$};
						\draw [fill=white] (0,{-\h*(2-1)}) circle (\r)  node { 2 } node (B) [left = 0.1]{ \footnotesize  $\boldsymbol{X}_2 \rightarrow$};
						
						\draw [line width = 1.5] (-0.2, 0.1) rectangle  (0.2, {-\h-0.1});
						\node (msg) at (-0.5, {-\h/2}) [rotate = 90] {\footnotesize$(W_1, W_2)$};
						\draw [->] (msg) -- (A);
						\draw [->] (msg) -- (B);
						
						\draw [fill=white] (\w,{-\h*(1-1)}) circle (\r)  node { 1} node [right = 0.1] { \footnotesize  $\rightarrow \boldsymbol{Y}_1 \rightarrow W_1, \xcancel{W_2}$ };
						\draw [fill=white] (\w,{-\h*(2-1)}) circle (\r)  node { 2} node [right = 0.1] { \footnotesize  $\rightarrow \boldsymbol{Y}_2\rightarrow W_2, \xcancel{W_1}$ };
						
						\draw [line width=1] (0+\r,{-\h*(1-1)} ) -- (\w-\r,{-\h*(1-1)}) node [pos = 0.5, above] {\footnotesize  $\alpha_{11} =\alpha$};
						\draw [line width=1] (0+\r,{-\h*(2-1)} ) -- (\w-\r,{-\h*(2-1)}) node [pos = 0.5, below] { \footnotesize  $\alpha_{22} = 1$};
						
						\draw [line width=1] (0+\r,{-\h*(2-1)} ) -- (\w-\r,{-\h*(1-1)}) node [pos = 0.4, right= 0.1] {\footnotesize  $\alpha_{12} = \beta$};
						
\end{tikzpicture}

				\caption{\it \small The Gaussian $Z$ Broadcast Channel (ZBC).}
				\label{fig:ZBC}
			\end{figure}
			
While our focus is primarily on the ZIC, as a useful point of reference let us also define the corresponding Gaussian $Z$ Broadcast Channel (ZBC), shown in Figure \ref{fig:ZBC}, which is identical to the ZIC in every regard except that in the ZBC the transmitters are allowed to cooperate fully to jointly encode the messages; i.e., $(\bX_1, \bX_2) = f_{0,n}(W_1, W_2,\theta_1,\theta_2)$, where $f_{0,n}$ is the encoding function. 

\subsection{The GDoF Framework}			
			Within the GDoF framework, the received signals in the $t$-th channel use are described as 
			\begin{align} 
				Y_1(t) &= G_{11}(t) \sqrt{P^{\alpha_{11}}} X_1(t) + G_{12}(t) \sqrt{P^{\alpha_{12}}} X_2(t) + Z_1(t), \label{eq:model1}\\
				Y_2(t) &= G_{22}(t) \sqrt{P^{\alpha_{22}}} X_2(t) + Z_2(t), \label{eq:model2}
			\end{align}
			where $P$ is a nominal  variable (referred to as \emph{power}) whose asymptotic limit, i.e., $P \rightarrow \infty$, will be used to define the GDoF metric. $Z_i(t), i=1,2,$ are the zero-mean unit-variance additive white Gaussian noise terms. $X_i(t), i=1,2,$ are the signals sent from the two transmitters, each of which is subject to a unit transmit power constraint. All symbols are real-valued. Without loss of generality,\footnote{There is no loss of generality in this assumption because from the definition of GDoF in \eqref{eq:GDoFdef} it is obvious that any normalization of $\alpha_{ij}$ parameters results in simply the same normalization factor appearing in the GDoF value.}
 let us normalize the $\alpha_{ij}$ parameters so that $\alpha_{22} = 1, \alpha_{12} = \beta$ and $\alpha_{11} = \alpha.$

			Let us briefly recall the motivation behind the GDoF framework. The channel strength parameters $\alpha_{ij}$ correspond (approximately) to the  capacity of the corresponding point to point Gaussian channel between Transmitter $j$ and Receiver $i$. Specifically, note that the links under the GDoF framework in (\ref{eq:model1}) and (\ref{eq:model2}) have approximate point-to-point capacities $\alpha_{ij}\Big(\frac{1}{2}\log (P)\Big)$. Here $\frac{1}{2}\log (P)$ may be viewed as a nominal scaling factor that is applied to proportionately scale the capacity of every link. The intuition behind this scaling is that as the capacity of every link is scaled by the same factor, the network capacity should scale  by approximately the same factor as well. Therefore,  normalizing all rates by $\frac{1}{2} \log (P)$ yields an approximation to the capacity of the network. Letting $P$ approach infinity makes the problem amenable to asymptotic analysis, which indeed gives us the definition of GDoF (See equation (\ref{eq:GDoFdef})). It is noteworthy that the deterministic models of \cite{Avestimehr_Diggavi_Tse}, which have been the key to numerous capacity approximations over the last decade, are specializations  of the GDoF framework under perfect CSIT.  For robust GDoF studies, however, we need to limit CSIT to finite precision.

\subsection{Finite Precision CSIT}
Following in the footsteps of \cite{Arash_Jafar}, let us define $\mathcal{G}$ as a set of random variables that satisfy the bounded density assumption of \cite{Arash_Jafar} (replicated as Definition \ref{def:bounded} in Section \ref{sec:def} of this paper). Elements of $\mathcal{G}$ may be viewed as random perturbation factors that are introduced into the model primarily to limit CSIT to finite precision,  thus their realizations are assumed to be known perfectly to the receivers but not to the transmitters. Formally,  
\begin{align}
I(W_1, W_2,\theta_1, \theta_2, \bX_1, \bX_2; \chg) &= 0. \label{eq:fp}
\end{align} 
Specifically, the channel coefficients 	$G_{ij}(t)$  are distinct elements of $\mathcal{G}$ for all $t \in [n], i = 1,2$. 

\subsection{Perfect CSIT}
While our focus in this work is primarily on finite precision CSIT, as a useful point of reference let us also introduce the perfect CSIT assumption, which implies that the channel coefficients $G_{ij}(t)$ are perfectly known not only to both receivers but to both transmitters as well. The constraint \eqref{eq:fp} does not hold under perfect CSIT, and the coding functions may depend on the channel realizations. Thus, $\bX_i = f_{i,n}(W_i, \theta_i, \mathcal{G})$, $i=1,2$ for the ZIC under perfect CSIT, and $(\bX_1, \bX_2) = f_{0,n}(W_1, W_2,\theta_1,\theta_2,\mathcal{G})$ for the ZBC under perfect CSIT. 

\subsection{Achievable Rates under Secrecy Constraint}
A  rate tuple $(R_1, R_2)$ is achievable subject to the secrecy constraint if, for all $\epsilon > 0$, there exist $n$-length codes for some $n>0$ such that (i) the size of each message set $|\mathcal{W}_i| \geq 2^{nR_i}$; (ii) the decoding error probabilities at both users are no larger than $\epsilon$; and (iii) the following secrecy constraint is satisfied
			\begin{align}
			\frac{1}{n}I(W_j; \boldsymbol{Y}_i\mid \chg)&\leq \epsilon&& \forall i,j \in\{ 1,2\}, i \neq j.\label{eq:secrecy}
			\end{align}
			The secure capacity region $\mathcal{C}_P$ is the closure of the set of all achievable secure rate tuples. 
\subsection{Secure GDoF Region}			
			The secure GDoF region $\mathcal{D}$ is defined as 
			\begin{align} \label{eq:GDoFdef}
			\mathcal{D} &\triangleq \left \{(d_1, d_2) \middle | \begin{array}{l}
			\hfill \forall i \in\{ 1,2\}\\
			\exists (R_1(P), R_2(P)) \in \mathcal{C}_P
			\end{array} ,
			d_i = \lim_{P \rightarrow \infty} \frac{R_i(P)}{\frac{1}{2}\log P}
			\right\}.
			\end{align}
We will use subscripts to distinguish ZIC from ZBC, and superscripts to distinguish finite precision CSIT from perfect CSIT, so for example, $\mathcal{D}_{\mbox{\tiny IC}}^{\tiny f.p.}$ symbolizes the GDoF region for the ZIC under finite precision CSIT, and $\mathcal{D}_{\mbox{\tiny BC}}^{\tiny p}$ is the GDoF region for the ZBC under perfect CSIT.

\section{Results} \label{sec:result}
In order to answer our titular question about the robustness of structured codes, we will compare the GDoF region of the ZIC under perfect CSIT with the GDoF region of the ZIC under finite precision CSIT, i.e., $\mathcal{D}_{\mbox{\tiny IC}}^{\tiny p}$ versus $\mathcal{D}_{\mbox{\tiny IC}}^{\tiny f.p.}$. These  are characterized below 	in Lemma \ref{lemma:gdof_perfect} and Theorem \ref{thm:gdof}, respectively.

\subsection{Secure GDoF of the ZIC with Perfect CSIT}
\begin{lemma}\label{lemma:gdof_perfect}
		The secure GDoF region of the ZIC under perfect CSIT is characterized as
		\begin{align}
			\mathcal{D}_{\mbox{\tiny IC}}^{\tiny p} = 
			\left \{ (d_1, d_2) \in \mathbb{R}_+^2 \middle | 
			\begin{array}{l}
				d_1 \leq \alpha \\
				d_2 \leq \min \{ 1, (1 + \alpha - \beta)^+ \} \\
				d_1 + d_2 \leq \alpha + (1 - \beta)^+
			\end{array}
			\right \}.
		\end{align}
\end{lemma}
While a direct statement of Lemma \ref{lemma:gdof_perfect} does not appear in prior literature to our knowledge, the lemma essentially follows from known results and arguments. For the sake of completeness, these arguments are summarized in Appendix \ref{sec:proofPerfect}.


	\subsection{Secure GDoF of the ZIC with Finite Precision CSIT}


	\begin{theorem} \label{thm:gdof}
		The secure GDoF region of the ZIC under finite precision CSIT is characterized as,
		\begin{enumerate}
			\item Regime 1: $1 < \beta < \alpha$
			\begin{align}
				\mathcal{D}_{\mbox{\tiny IC}}^{\tiny f.p.} = \left\{
					(d_1, d_2) \in \mathbb{R}^{2}_+ \middle | \begin{array}{l}
						d_2 \leq 1,\\
						d_1 + \beta d_2 \leq \alpha
					\end{array}
				\right\}.\label{thm:gdof-1}
			\end{align}
			\item Regime 2: $1 < \beta$ and $\beta-1 < \alpha \leq \beta$
			\begin{align}
				\mathcal{D}_{\mbox{\tiny IC}}^{\tiny f.p.} = \left\{
					(d_1,d_2) \in \mathbb{R}^{2}_+ \middle |~ \frac{d_1}{\alpha} + \frac{d_2}{1+\alpha - \beta} \leq 1
				\right\}. \label{thm:gdof-2}
			\end{align}
			\item Regime 3: $1 < \beta$ and $\alpha \leq \beta-1$
			\begin{align}
			\mathcal{D}_{\mbox{\tiny IC}}^{\tiny f.p.} &=	\left\{  (d_1, d_2)  \in \mathbb{R}^{2}_+ \middle |~ d_1 \leq \alpha, d_2 = 0 \right\}. \label{thm:gdof-3ic}
			\end{align} 
			\item Regime 4: $0 \leq \beta \leq 1$
			\begin{align}
				\mathcal{D}_{\mbox{\tiny IC}}^{\tiny f.p.} = \left\{ (d_1, d_2)   \in \mathbb{R}^{2}_+ \middle | 
					\begin{array}{l}
						d_1 \leq \alpha, d_2 \leq 1 \\
						d_1 + d_2 \leq 1 + \alpha - \beta
					\end{array}
				\right \}. \label{thm:gdof-4}
			\end{align}
		\end{enumerate}
\end{theorem}

The proof of Theorem \ref{thm:gdof} appears in Section \ref{sec:proof} and \ref{sec:converse}. 
The main contribution of this work is the proof of Theorem \ref{thm:gdof} for Regimes $1$ and $2$. Indeed, Regime $3$ is  trivial and Regime $4$ already follows from \cite{Chan_Geng_Jafar_secureBC}. The converse  proofs for Regimes $1$ and $2$ rely on various sum-set inequalities of \cite{Arash_Jafar_sumset}, and are  central to the thesis of this work, that the benefits of structured jamming are not robust to finite precision CSIT in the GDoF sense.

%

\subsection{How Robust are Structured Codes?}
With the help of Lemma \ref{lemma:gdof_perfect} and Theorem \ref{thm:gdof}, we are ready to explore the robustness of the GDoF gains from structured codes through the following observations. 

\begin{enumerate}[leftmargin=*]
		\begin{figure}[t]
	\centering
		\begin{tikzpicture}[scale=2.75]
		\def \e {0.0}
		\draw [line width = 0.5, ->] (0,0) node [below left] {\footnotesize $O$} -- (2.1,0) node [right] {\footnotesize $\alpha$};
		\draw [line width = 0.5, ->] (0,0) -- (0,2.1) node [above] {\footnotesize $\beta$};
		
		\draw [line width = 0.5, dashed] (0,0) -- (2.1,2.1);
		\draw [line width = 0.5, dashed] (0,1) -- (1.1, 2.1);
		\draw [line width = 0.5, dashed] (0,1) -- (2.1, 1);
		
		\draw [line width = 0, fill = blue!20]  (2, {1+\e})--({1+2*\e}, {1+\e}) -- ({2, 2-\e}); 
		\draw [line width = 1] (2.1, {1+\e})--( {1+2*\e}, {1+\e}) -- (2.1, {2.1-\e}); 
		\draw [line width = 0, fill = blue!50]  ({2-\e}, 2) -- ({1-\e}, {1+\e})--( {2*\e}, {1+\e}) -- ({1+\e}, 2) ; 
		\draw [line width = 1]  ({2.1-\e}, 2.1) -- ({1-\e}, {1+\e})--({2*\e}, {1+\e}) -- ({1.1+\e}, 2.1); 
		\draw [line width = 0, fill = purple!20]  ( {1-\e},2) -- (\e, {1+2*\e} )--(\e, 2) ; 
		\draw [line width = 1]  (1.1-\e, 2.1) -- (\e, {1+2*\e})--(\e, 2.1); 
		\draw [line width = 0, fill = purple!50]  (2, \e) -- (\e, \e)-- ({\e}, {1-\e}) -- (2, {1-\e}) ; 
		\draw [line width = 1]  ( 2.1, \e) -- (\e, \e)-- ({\e}, {1-\e}) -- (2.1, {1-\e}) ; 
		
				\draw [line width = 0.25, dashed] (1,0) -- (1,1);

		\node at (1.69, 1.2) {\footnotesize Regime 1};
		\node at (1, 1.5) {\footnotesize Regime 2};
		\node at (0.3, 1.75) {\footnotesize Regime 3};
		\node at (1, 0.5)[fill=purple!50] {\footnotesize Regime 4};
		\node at (-0.1, 1) {\footnotesize $1$}; 
		\node at (-0.1, 2) {\footnotesize $2$};
		\node at (1, -0.1) {\footnotesize $1$};
		\node at (2, -0.1) {\footnotesize $2$};   
	\end{tikzpicture}

	\caption{\small The parameter regimes corresponding to the four cases in Theorem \ref{thm:gdof}.}
	\label{fig:regime}
\end{figure}
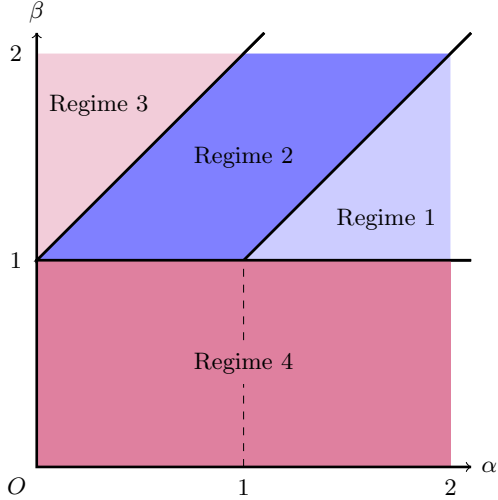
\item There are $4$ parameter regimes identified in Theorem \ref{thm:gdof}. These regimes are shown in Figure \ref{fig:regime}. Our first observation is that in regimes $3$ and $4$, we have $\mathcal{D}_{\mbox{\tiny IC}}^{\tiny p}=\mathcal{D}_{\mbox{\tiny IC}}^{\tiny f.p.}$, i.e., there is no loss of GDoF from limiting CSIT to finite precision.  However, this is not because structured codes are robust to finite precision CSIT. Upon  inspection of the achievable scheme, it is evident that these are the regimes where structured codes are not needed even with perfect CSIT. In  Regime $3$ we only need to switch off Transmitter $2$, thus allowing User $1$ to achieve $\alpha$ GDoF. It is not possible for User $2$ to achieve any positive GDoF value in Regime $3$ without violating the secrecy constraint because the signal from Transmitter $2$ appears at Receiver $1$ with so much strength ($\beta\geq \alpha+1$), that even if Transmitter $1$ uses all its power to only transmit noise, thus maximally elevating the noise floor at Receiver $1$, the  interfering signal that appears above the noise floor at Receiver $1$ still reveals everything that is visible to Receiver $2$. In Regime $4$ (see \cite{Chan_Geng_Jafar_secureBC}) all we need is for Transmitter $1$ to transmit enough noise (jamming) to elevate the noise floor at Receiver $1$ to the level of the interfering signal, and then send its desired message above the new noise floor. The jamming guarantees security, and the desired signal is decoded by Receiver $1$ simply by treating everything else as noise. Thus, there is no need for structured codes to allow alignment or aggregate decoding of signals.

	\begin{figure}[t]
		\centering
		\subcaptionbox{}{ 
			\centering
			\begin{tikzpicture}[yscale = 2, xscale=2.7]
					\def \e {0.01}
					
					\draw [line width = 2, ->] (-0.1, 0) node [below left] {$O$} -- (2, 0) node [right] {$d_1$};
					\draw [line width = 2, ->] (0, -0.1) -- (0, 2) node [right] {$d_2$};

					\draw [draw = red, fill = red!30, line width =1] (0,0.75) node [left] {$1$} -- (0.75,0.75) -- (1.5,0) -- (0,0) -- cycle;
					\draw [draw = black, fill = gray!30, line width =1] (\e,{0.75-\e}) -- (0.25,{0.75-\e}) -- ({1.5-\e},\e) -- (\e,\e) -- cycle;
					
					\draw  [line width = 0.5, dashed] (1.5,0) node [below] {$\alpha$} -- (0,1.5) node [left] {$\alpha$};
					\draw [line width = 0.5, dashed] (0.25, 0.75) -- (0.25, 0) node [below] {$\alpha-\beta$};
					\draw [line width = 0.5, dashed] (0.75, 0.75) -- (0.75, 0) node [below] {$\alpha-1$};
					
				\end{tikzpicture}
			} 
		\hfill
		\subcaptionbox{} { 
			\centering
			\begin{tikzpicture}[scale=0.8]
			\foreach \m in {1,2}
			{
				\coordinate (M\m) at (1,1.5-2*\m);
				\coordinate (N\m) at (5,1.5-2*\m){};
			};
			
			\draw [very thick] (M1)--(N1) node [pos=0.7, above = 0.1cm] {$\alpha$};
			\draw[very thick] (M2)--(N2) node [pos=0.7, below = 0 cm, text  = black] { $1$};
			
			\draw[very thick, red] (M2)--(N1) node [pos=0.7, below = 0cm, text=black] {$\beta$};

			\node[thick, circle, draw=black, fill=white, inner sep = 1.5, left] at  (M1) {\footnotesize $\bX_1$};
			\node[thick, circle, draw=black, fill=white, inner sep = 1.5, right] at (N1){\footnotesize $\bY_1$};
			
			\node[thick, circle, draw=black, fill=white, inner sep = 1.5, left] at  (M2) {\footnotesize $\bX_2$};
			\node[thick, circle, draw=black, fill=white, inner sep = 1.5, right] at (N2){\footnotesize $\bY_2$};

			\draw  [fill=red!20!white](-0.75,1) rectangle (-0.25, 1.5) node[label={[xshift=-0.6cm, yshift=-0.5cm]}, pos=0.5] {};
			\draw  [pattern=dots, pattern color=red](-0.75,0) rectangle (-0.25, 1) node[label={[xshift=-0.9cm, yshift=-0.7cm]}] {};
			\path [thick] (-1,0.5)--(0,0.5);
			\draw  [fill=red!50!white](-0.75,-0.5) rectangle (-0.25, 0) node[label={[xshift=-0.6cm, yshift=-0.5cm]}, pos=0.5] {};
			\draw[<->, thick](-0.9,1)--(-0.9,1.5) node[left, midway]{\footnotesize $\alpha-\beta$};
			\draw[<->, thick](-0.9,-0.5)--(-0.9,0) node[left, midway]{\footnotesize $\beta-1$};
			\draw[<->, thick](-0.1,-0.5)--(-0.1,1.5) node[right, midway]{\footnotesize $\alpha$};
			\draw [thick](-1,-0.5)--(0,-0.5) ;
			
			\begin{scope}[shift={(0.5,0)}]
				\draw  [fill=red!20!white](6,1) rectangle (6.5, 1.5) node[label={[xshift=-0.87cm, yshift=-0.7cm]}]  {};;

				\draw  [pattern=dots, pattern color=purple](6,0) rectangle (6.5, 1) node[label={[xshift=-0.87cm, yshift=-0.7cm]}]  {};
				\draw  [fill=red!50!white](6,-0.5) rectangle (6.5, 0) node[label={[xshift=-0.63cm, yshift=-0.57cm]}, pos=0.5]{};
				\draw  [pattern=dots, pattern color=blue](6.5,0) rectangle (7, 1) node[label={[xshift=0.3cm, yshift=-0.7cm]}]  {};
				\draw  [fill=white](6.5,-0.5) rectangle (7, 0) node[label={[xshift=0.3cm, yshift=-0.7cm]}]  {};
				\draw[<->, thick](7.15,-0.5)--(7.15,1) node[right, midway]{\footnotesize $\beta$};
				\draw[<->, thick](5.85,-0.5)--(5.85,1.5) node[left, midway]{\footnotesize $\alpha$};
				\draw [thick](5.5,-0.5)--(7.5,-0.5) ;
			\end{scope}
		
			\draw  [pattern=dots, pattern color=blue](-0.75,-2) rectangle (-0.25, -1) node[label={[xshift=-0.9cm, yshift=-0.7cm]}]  {};
			\draw  [fill=white](-0.75,-2.5) rectangle (-0.25, -2) node[label={[xshift=-0.9cm, yshift=-0.7cm]}]  {};
			\draw[<->, thick](-0.1,-2.5)--(-0.1,-1) node[right, midway]{\footnotesize $\beta$};
			\draw[<->, thick](-0.9,-2)--(-0.9,-1) node[left, midway]{\footnotesize $1$};

			\draw [thick](-1,-2.5)--(0,-2.5);
			
			\begin{scope}[shift={(0.3,0)}]
				\draw  [pattern=dots, pattern color=blue](6.5,-2.5) rectangle (7, -1.5) node[label={[xshift=0.3cm, yshift=-0.7cm]}]  {};
				\draw[<->, thick](7.15,-1.5)--(7.15,-2.5) node[right, midway]{\footnotesize $1$};
				\draw [thick](5.75,-2.5)--(7.75,-2.5);
			\end{scope}	
			\end{tikzpicture}
			}
		\caption{\small  (a) $\mathcal{D}_{\mbox{\tiny IC}}^{\tiny p}$ (in red) and $\mathcal{D}_{\mbox{\tiny IC}}^{\tiny f.p.}$ (in grey) are shown for Regime $1$ (where $1<\beta<\alpha$). (b) The achievability of $(d_1^{**},d_2^*)=(\alpha-1,1)$ under perfect CSIT is illustrated. In particular, aggregate decoding and cancellation of lattice-aligned signals (blue and red dotted portions) is required, which is only possible under perfect CSIT. Signal levels shown in plain white are empty.}
		\label{fig:gdof1}
		
	\end{figure}
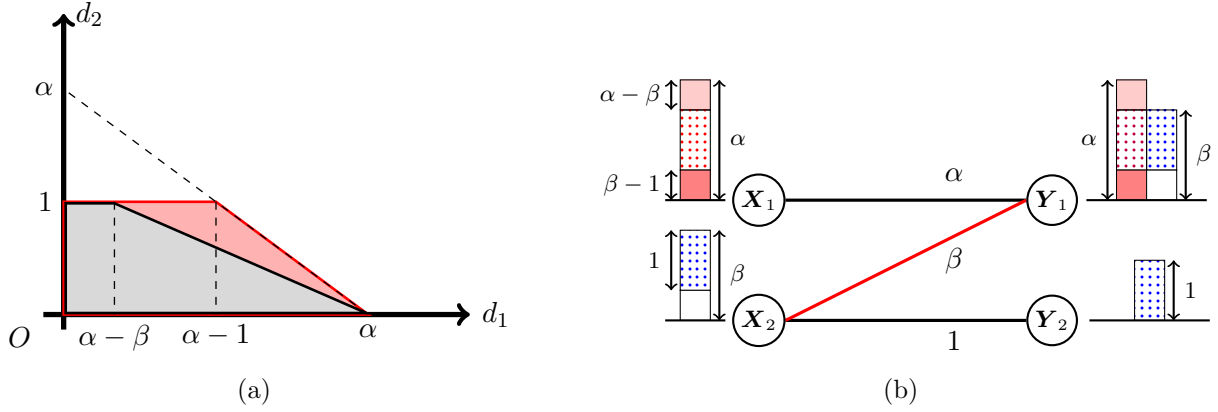

\item In regimes $1$ and $2$ a gap appears between $\mathcal{D}_{\mbox{\tiny IC}}^{\tiny p}$ and $\mathcal{D}_{\mbox{\tiny IC}}^{\tiny f.p.}$. Indeed, these regimes are central to this work, as they reveal the fragility of structured codes. First let us consider Regime $1$. The GDoF regions, $\mathcal{D}_{\mbox{\tiny IC}}^{\tiny p}$ and $\mathcal{D}_{\mbox{\tiny IC}}^{\tiny f.p.}$ for this regime are illustrated in Figure \ref{fig:gdof1}(a). Let $d_2^*$ denote the maximal value of $d_2$. According to  Figure \ref{fig:gdof1}(a), $d_2^*=1$. Conditioned on $d_2=d_2^*$, let $d_1^{**}$ denote the maximum value of $d_1$.  We note that under perfect CSIT we have $(d_1^{**}, d_2^*)=(\alpha-1,1)$ but under finite precision CSIT we only have $(d_1^{**}, d_2^*)=(\alpha-\beta,1)$. This loss of GDoF reveals the fragility of aggregate decoding of structured codes. For an intuitive explanation, consider Figure \ref{fig:gdof1}(b) which shows how $(d_1^{**}, d_2^*)=(\alpha-1,1)$ is achieved under perfect CSIT, by lattice alignment between the dotted portions of signals seen at Receiver $1$. This lattice alignment ensures the secrecy of $W_2$ from Receiver $1$, while simultaneously allowing Receiver $1$ to  decode the sum of lattice points as a valid codeword. Indeed, while the top $\alpha-\beta$ GDoF (shown in light red) of desired message can be decoded by Receiver $1$ without any need for alignment, it is the aggregate decoding of aligned signals that allows Receiver $1$ to decode the additional bottom $\beta-1$ GDoF (shown in dark red) of desired message, thus achieving a total of $d_1^{**}=(\alpha-\beta)+(\beta-1)=\alpha-1$ GDoF. Intuitively, under finite precision CSIT, aggregate decoding and cancellation are not possible, thus Receiver $1$ is only able to decode the top $\alpha-\beta$ GDoF of desired message, i.e., $d_1^{**}=\alpha-\beta$. The main technical challenge  in this work	is to prove this intuition, i.e., to show that aggregate decoding or any other structured jamming scheme that even partially retains  the GDoF benefits of aggregate decoding and cancellation, is not possible under finite precision CSIT.
	
	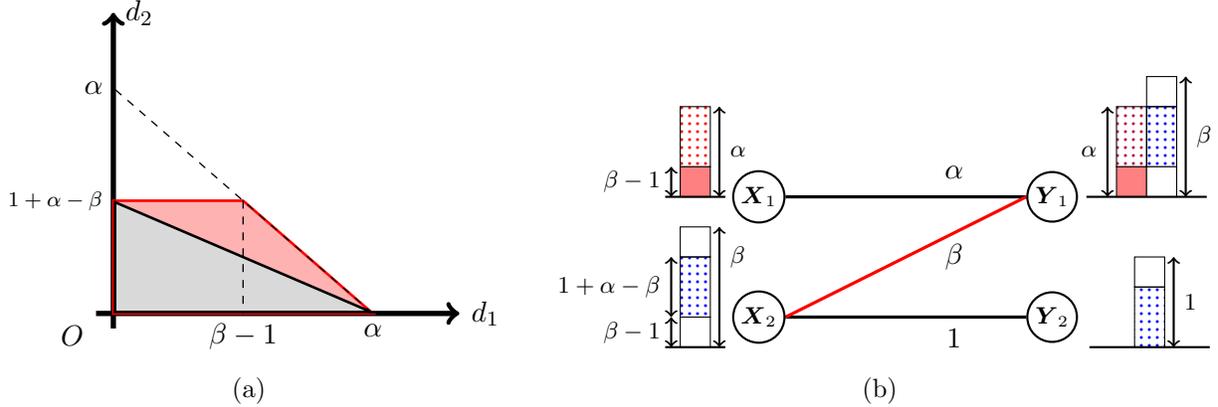
\begin{figure}[t]
	\centering
	\subcaptionbox{}{ 
		\centering
		\begin{tikzpicture}[yscale = 2, xscale=2.3]
			\def \e {0.01}
				
			\draw [line width = 2, ->] (-0.1, 0) node [below left] {$O$} -- (2, 0) node [right] {$d_1$};
			\draw [line width = 2, ->] (0, -0.1) -- (0, 2) node [right] {$d_2$};

			\draw [draw = red, fill = red!30, line width =1] (0,0.75) node [left] {\scriptsize$1+\alpha-\beta$} -- (0.75,0.75) -- (1.5,0) -- (0,0) -- cycle;
			\draw [draw = black, fill = gray!30, line width =1] (\e,{0.75-\e}) -- ({1.5-\e},\e) -- (\e,\e) -- cycle;
				
			\draw  [line width = 0.5, dashed] (1.5,0) node [below] {$\alpha$} -- (0,1.5) node [left] {$\alpha$}; 
			\draw [line width = 0.5, dashed] (0.75, 0.75) -- (0.75, 0) node [below] {$\beta-1$};
				
		\end{tikzpicture}
		}
	\hfill
	\subcaptionbox{} { 
		\centering
		\begin{tikzpicture}[scale=0.8]
			\foreach \m in {1,2}
			{
				\coordinate (M\m) at (1,1.5-2*\m);
				\coordinate (N\m) at (5,1.5-2*\m){};
			};
	
			\draw [very thick] (M1)--(N1) node [pos=0.7, above = 0.1cm] {$\alpha$};
			\draw[very thick] (M2)--(N2) node [pos=0.7, below = 0 cm, text  = black] { $1$};
	
			\draw[very thick, red] (M2)--(N1) node [pos=0.7, below = 0cm, text=black] {$\beta$};

			\node[thick, circle, draw=black, fill=white, inner sep = 1.5, left] at  (M1) {\footnotesize $\bX_1$};
			\node[thick, circle, draw=black, fill=white, inner sep = 1.5, right] at (N1){\footnotesize $\bY_1$};
	
			\node[thick, circle, draw=black, fill=white, inner sep = 1.5, left] at  (M2) {\footnotesize $\bX_2$};
			\node[thick, circle, draw=black, fill=white, inner sep = 1.5, right] at (N2){\footnotesize $\bY_2$};

			\draw  [pattern=dots, pattern color=red](-0.75,0) rectangle (-0.25, 1) node[label={[xshift=-0.9cm, yshift=-0.7cm]}] {};
			\path [thick] (-1,0.5)--(0,0.5);
			\draw  [fill=red!50!white](-0.75,-0.5) rectangle (-0.25, 0) node[label={[xshift=-0.6cm, yshift=-0.5cm]}, pos=0.5] {};
			\draw[<->, thick](-0.1,-0.5)--(-0.1,1) node[right, midway]{\footnotesize $\alpha$};
			\draw[<->, thick](-0.9,-0.5)--(-0.9,0) node[left, midway]{\footnotesize $\beta-1$};
			\draw [thick](-1,-0.5)--(0,-0.5) ;
	
			\begin{scope}[shift={(0.5,0)}]
				\draw  [pattern=dots, pattern color=purple](6,0) rectangle (6.5, 1) node[label={[xshift=-0.87cm, yshift=-0.7cm]}]  {};
				\draw  [fill=red!50!white](6,-0.5) rectangle (6.5,0) node[label={[xshift=-0.87cm, yshift=-0.7cm]}]  {};;
				\draw  [fill=white](6.5,1) rectangle (7, 1.5) node[label={[xshift=0.3cm, yshift=-0.7cm]}]  {};
				\draw  [pattern=dots, pattern color=blue](6.5,0) rectangle (7, 1) node[label={[xshift=0.3cm, yshift=-0.7cm]}]  {};
				\draw  [fill=white](6.5,-0.5) rectangle (7, 0) node[label={[xshift=0.3cm, yshift=-0.7cm]}]  {};
				
				\draw[<->, thick](7.15,-0.5)--(7.15,1.5) node[right, midway]{\footnotesize $\beta$};
				\draw[<->, thick](5.85,-0.5)--(5.85,1) node[left, midway]{\footnotesize $\alpha$};
				\draw [thick](5.5,-0.5)--(7.5,-0.5) ;
			\end{scope}
	
			\draw  [fill=white](-0.75,-1.5) rectangle (-0.25, -1) node[label={[xshift=-0.9cm, yshift=-0.7cm]}]  {};
			\draw  [pattern=dots, pattern color=blue](-0.75,-2.5) rectangle (-0.25, -1.5) node[label={[xshift=-0.9cm, yshift=-0.7cm]}]  {};
			\draw  [fill=white](-0.75,-3) rectangle (-0.25, -2.5) node[label={[xshift=-0.9cm, yshift=-0.7cm]}]  {};	
			\draw[<->, thick](-0.9,-2.5)--(-0.9,-1.5) node[left, midway]{\footnotesize $1+\alpha-\beta$};
			\draw[<->, thick](-0.9,-3)--(-0.9,-2.5) node[left, midway]{\footnotesize $\beta-1$};
			\draw[<->, thick](-0.1,-3)--(-0.1,-1) node[right, pos = 0.75]{\footnotesize $\beta$};
			
			\draw [thick](-1,-3)--(0,-3);
	
			\begin{scope}[shift={(0.3,0)}]
				\draw  [fill=white](6.5,-2) rectangle (7, -1.5) node[label={[xshift=0.3cm, yshift=-0.7cm]}]  {};
				\draw  [pattern=dots, pattern color=blue](6.5,-3) rectangle (7, -2) node[label={[xshift=0.3cm, yshift=-0.7cm]}]  {};
				\draw[<->, thick](7.15,-1.5)--(7.15,-3) node[right, midway]{\footnotesize $1$};
				\draw [thick](5.75,-3)--(7.75,-3);
			\end{scope}	
		\end{tikzpicture}
	}
		
	\caption{\small  (a) $\mathcal{D}_{\mbox{\tiny IC}}^{\tiny p}$ (in red) and $\mathcal{D}_{\mbox{\tiny IC}}^{\tiny f.p.}$ (in grey) are shown for Regime $2$ (where $1<\beta$ and $\beta-1<\alpha\leq \beta$). (b) The achievability of $(d_1^{**},d_2^*)=(\beta-1,1+\alpha-\beta)$ under perfect CSIT is illustrated. In particular, aggregate decoding of lattice-aligned signals (blue and red dotted portions) is required, which is only possible under perfect CSIT. Signal levels shown in plain white are empty.}
	\label{fig:gdof2}
	\end{figure}

\item Now let us consider Regime $2$, for which the GDoF regions $\mathcal{D}_{\mbox{\tiny IC}}^{\tiny p}$ and $\mathcal{D}_{\mbox{\tiny IC}}^{\tiny f.p.}$  are illustrated in Figure \ref{fig:gdof2}(a). In this case the loss of GDoF is even more severe as we have $(d_1^{**}, d_2^*)=(\beta-1,1+\alpha-\beta)$ under perfect CSIT, and only $(d_1^{**}, d_2^*)=(0,1+\alpha-\beta)$ under finite precision CSIT. The loss of GDoF is once again attributable to the fragility of aggregate decoding, as illustrated in Figure \ref{fig:gdof2}(b). Aggregate decoding and cancellation of lattice-aligned signals allows Receiver $1$ to decode the bottom $\beta-1$ GDoF of desired message under perfect CSIT, thus achieving $d_1^{**}=\beta-1$. Intuitively, under finite precision CSIT, Receiver $1$ is no longer able to decode the aggregate signal, indeed $d_1^{**}=0$. Once again, the challenge is to formalize and prove this intuition, for which we will rely on sum-set inequalities of \cite{Arash_Jafar_sumset}.

	\begin{figure}[t]
		\centering
		\tdplotsetmaincoords{60}{330} 
		\tikzstyle{fpFacet} = [{blue, fill opacity=0.5, draw=blue!50!black, draw opacity = 1, line width = 0.5, rounded corners=0.5pt}]
		\tikzstyle{pFacet} = [{red, fill opacity=0.25, draw=red, draw opacity = 1, line width = 0.5, rounded corners=0.5pt}]
		
			\begin{tikzpicture}[scale = 1.5, tdplot_main_coords]
			\def \e {0.0}
			\draw[gray, thick, ->, >=stealth'] (0,0,0) -- (2.5,0,0) node [black, above ] {$\alpha$};
			\draw[gray, thick, ->, >=stealth'] (0,0,0) -- (0,2.5,0) node [black, above ] {$\beta$};
			\draw[gray, thick, ->, >=stealth'] (0,0,0) -- (0,0,2.5) node [black, above  ] {$d_1^{**}$};
			\draw[black, dashed, line width = 0.5] (0,1,0) -- (2,1,0);
			\draw[black, dashed, line width = 0.5] (0,2,0) -- (2,2,0);
			\draw[black, dashed, line width = 0.5] (1,0,0) -- (1,2,0);
			\draw[black, dashed, line width = 0.5] (2,0,0) -- (2,2,0);
			\draw[black, dashed, line width = 0.5] (2,1,0) -- (2,1,1);
			\draw[black, dashed, line width = 0.5] (1,2,0) -- (1,2,1);
			\draw[black, dashed, line width = 0.5] (2,2,0) -- (2,2,1);
			\draw[black, dashed, line width = 0.5] (1,2,0) -- (1,2,1);

			\filldraw [fpFacet, fill=blue!80!black, opacity=0.2, thick] (0,2,{-\e}) -- (1, 2, {-\e}) -- (1,2,{1-\e}) -- cycle; 
			\filldraw [fpFacet, fill=blue!95!black, opacity=0.6] (0,1,{-\e}) -- (0, 2, {-\e}) -- (1,2,{1-\e}) -- cycle; 
			\filldraw [fpFacet, fill=blue!80!black, opacity=0.5, thick] (0,1,{-\e}) -- (1, 2, {-\e}) -- (1,2,{1-\e}) -- cycle; 

			\filldraw [fpFacet, fill=blue!70!black, opacity=0.3] (0,1,{-\e}) --  (1,1,{-\e}) -- (2,2,{-\e}) -- (1, 2, {-\e}); 

			\filldraw [fpFacet, fill=blue!70!black, opacity=0.2,  thick] (2,1,{-\e}) -- (2, 2, {-\e}) -- (2,1,{1-\e}) -- cycle; 
			\filldraw [fpFacet, fill=blue!80!black, opacity=0.5] (1,1,{-\e}) -- (2, 2, {-\e}) -- (2,1,{1-\e}) -- cycle; 
			\filldraw [fpFacet, fill=blue!95!black, opacity=0.6, thick] (1,1,{-\e}) -- (2, 1 {-\e}) -- (2,1,{1-\e}) -- cycle; 

			\node at (0,0,0) [below left ] {$O$};
			\node at (1,0,0) [below right ] {$1$};
			\node at (2,0,0) [below right ] {$2$};
			\node at (0,1,0) [below left ] {$1$};
			\node at (0,2,0) [below left ] {$2$};
			\node at (2,1,1) [right]{\footnotesize $(2,1,1)$};
			\node at (1,2,1) [left]{\footnotesize $(1,2,1)$};
			
			\begin{scope}[shift={(4,-2.3,0)}]
				\draw[gray, thick, ->, >=stealth'] (0,0,0) -- (2.5,0,0) node [black, above ] {$\alpha$};
				\draw[gray, thick, ->, >=stealth'] (0,0,0) -- (0,2.5,0) node [black, above ] {$\beta$};
				\draw[gray, thick, ->, >=stealth'] (0,0,0) -- (0,0,2.5) node [black, above  ] {$d_1^{**}$};
				\draw[black, dashed, line width = 0.5] (0,1,0) -- (2,1,0);
				\draw[black, dashed, line width = 0.5] (0,2,0) -- (2,2,0);
				\draw[black, dashed, line width = 0.5] (1,0,0) -- (1,2,0);
				\draw[black, dashed, line width = 0.5] (2,0,0) -- (2,2,0);
				\draw[black, dashed, line width = 0.5] (2,1,0) -- (2,1,1);
				\draw[black, dashed, line width = 0.5] (1,2,0) -- (1,2,1);
				\draw[black, dashed, line width = 0.5] (2,2,0) -- (2,2,1);
				\draw[black, dashed, line width = 0.5] (1,2,0) -- (1,2,1);

				\filldraw [pFacet, fill=red!90!black, opacity=0.2, thick] (0,2,{-\e}) -- (2, 2 {-\e}) -- (2,2,{1-\e})-- (1,2,{1-\e}) -- cycle; 
	
				\filldraw [pFacet, fill=red!70!black, opacity=0.7] (0,1,0) -- (0, 2,0) -- (1,2,1) -- cycle; 
				\filldraw [pFacet, fill=red!99!black, opacity=0.6] (0,1,0) -- (1, 2,1) -- (2,2,1) -- (1,1,0)-- cycle; 

				\filldraw [pFacet, fill=red!90!black, opacity=0.2, thick] (2,1,{-\e}) -- (2, 2 {-\e}) -- (2,2,{1-\e})-- (2,1,{1-\e}) -- cycle; 
	
				\filldraw [pFacet, fill=red!85!black, opacity=0.6] (1,1,0) -- (2, 2,1) -- (2,1,1) -- cycle; 
				\filldraw [pFacet, fill=red!70!black, opacity=0.7, thick] (1,1,{-\e}) -- (2, 1 {-\e}) -- (2,1,{1-\e}) -- cycle; 
	
				\node at (0,0,0) [below left ] {$O$};
				\node at (1,0,0) [below right ] {$1$};
				\node at (2,0,0) [below right ] {$2$};
				\node at (0,1,0) [below left ] {$1$};
				\node at (0,2,0) [below left ] {$2$};
				\node at (2,1,1) [right]{\footnotesize $(2,1,1)$};
				\node at (2,2,1) [above]{\footnotesize $(2,2,1)$};
				\node at (1,2,1) [left]{\footnotesize $(1,2,1)$};
			\end{scope}			
			
\end{tikzpicture}

		\caption{\small  $d_1^{**}$  under finite precision CSIT (blue) and perfect  CSIT (red) in the parameter regimes $1,2,3$. Regime $4$ is omitted. Peak vertices are labeled as $(\alpha, \beta, d_1^{**})$ tuples.} 
		\label{fig:corner}
	\end{figure}
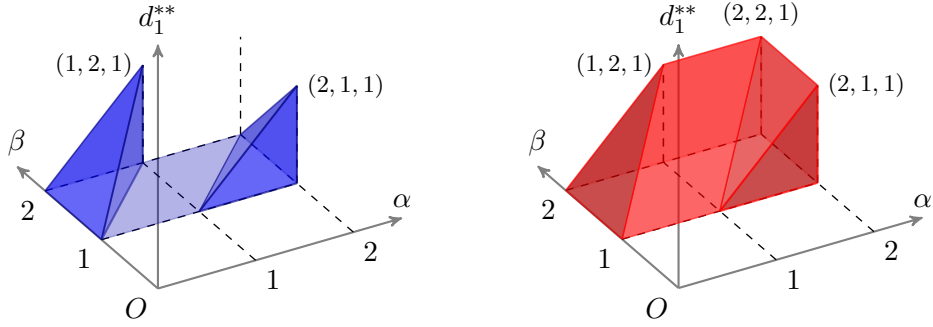
\item The loss of GDoF in terms of $d_1^{**}$ values is illustrated for the entirety of Regimes $1,2,3$ in Figure \ref{fig:corner}. As noted, there is no loss in Regime $3$, and Regime $4$ is omitted to avoid clutter.  Regime $2$ is particularly striking because $d_1^{**}=0$ under finite precision CSIT.  The discontinuity between Regime $2$ and Regime $3$ is interesting, because it shows the tremendous cost for  securing $W_2$ that is incurred in Regime $2$ where $d_2^*>0$. Note that this cost  disappears in Regime $3$ where $d_2^*=0$.

\item While the previous observations emphasized the loss of GDoF, let us now provide a counterpoint to show that the loss is bounded. As another measure of the loss of GDoF, consider an arbitrary weighted sum of GDoF values, say $d(w_1,w_2)=w_1d_1+w_2d_2$. Let us denote the maximal value of $d(w_1,w_2)$ for the ZIC under finite precision CSIT as $d^{f.p.}_{\mbox{\tiny IC}}(w_1,w_2) = \max_{(d_1,d_2)\in \mathcal{D}_{\mbox{\tiny IC}}^{\tiny f.p.}}w_1d_1+w_2d_2$. Similarly, for perfect CSIT we have $d^{p}_{\mbox{\tiny IC}}(w_1,w_2) = \max_{(d_1,d_2)\in \mathcal{D}_{\mbox{\tiny IC}}^{\tiny p}}w_1d_1+w_2d_2$. Based on Lemma \ref{lemma:gdof_perfect} and Theorem \ref{thm:gdof}, it is not difficult to verify that the extremal value, 
\begin{align}
\inf_{(\alpha,\beta)\in\mathbb{R}_2^+}\inf_{(w_1,w_2)\in\mathbb{R}_2^+}\frac{d^{f.p.}_{\mbox{\tiny IC}}(w_1,w_2)}{d^{p}_{\mbox{\tiny IC}}(w_1,w_2)}&=\frac{1}{2}.
\end{align}
In other words, looking out from the origin, the GDoF region $\mathcal{D}_{\mbox{\tiny IC}}^{\tiny f.p.}$ is \emph{at least} half as large in every direction as the GDoF region $\mathcal{D}_{\mbox{\tiny IC}}^{\tiny p}$. It is also easy to see that the bound is asymptotically tight because, e.g., in Figure \ref{fig:gdof1}(a), if we let $\beta\rightarrow \alpha$ from below and $\alpha\rightarrow\infty$, then $\mathcal{D}_{\mbox{\tiny IC}}^{\tiny p}$ approaches an almost-rectangular shape (with vertices $(0,0), (\alpha,0), (\alpha-1,1), (0,1)$) and $\mathcal{D}_{\mbox{\tiny IC}}^{\tiny f.p.}$ approaches the lower left half triangle created by a diagonal-wise partitioning of the rectangle (with vertices $(0,0), (\alpha,0), (0,1)$). Looking out along the other diagonal (the ray that passes through the origin and $(\alpha-1,1)$) we note that $\mathcal{D}_{\mbox{\tiny IC}}^{\tiny f.p.}$ is (asymptotically) only  half as large as $\mathcal{D}_{\mbox{\tiny IC}}^{\tiny p}$. Note that this corresponds to $(w_1,w_2)=(\alpha-1,1)$.
\end{enumerate}

\subsection{Secure GDoF of the ZBC with Perfect and Finite Precision CSIT}
The ZBC setting is less of our focus because even under perfect CSIT, the ZBC does not require lattice codes or aggregate decoding and cancellation of jammed signals for secure communication. Instead, it achieves secure communication through zero-forcing, which is conceptually much more straightforward. Nevertheless, it is also not robust under channel uncertainty. Moreover, the loss of GDoF in the ZBC under finite precision CSIT is also implied, as a byproduct of our analysis of the ZIC. This is because, remarkably, our converse proofs for Regimes $1,2$ in Theorem \ref{thm:gdof} hold even if we allow full cooperation among transmitters. Therefore, as our final result let us present the GDoF characterization of the ZBC under both perfect and finite precision CSIT.

\begin{theorem}\label{thm:gdofBC}
	The secure GDoF region of the ZBC under perfect CSIT, $\mathcal{D}_{\mbox{\tiny BC}}^{\tiny p}$ and under finite precision CSIT, $\mathcal{D}_{\mbox{\tiny BC}}^{\tiny f.p.}$, are characterized as 
	\begin{align}
		\mathcal{D}_{\mbox{\tiny BC}}^{\tiny p} &= \lrbr{ (d_1, d_2,) \in \mathbb{R}^2_+ \middle | d_1 \leq \max\{ \alpha, \beta-1 \}, ~d_2 \leq (1 - (\beta-\alpha)^+)^+ },\\
		\mathcal{D}_{\mbox{\tiny BC}}^{\tiny f.p.}  &= 
		\begin{cases}
			\lrbr{(d_1, d_2) \in \mathbb{R}^2_+ \middle | d_1 \leq \beta-1, d_2=0} & \text{if } 1 < \beta \text{ and } \alpha \leq \beta-1,\\
			\mathcal{D}_{\mbox{\tiny IC}}^{f.p.} & \text{otherwise.}
		\end{cases}
	\end{align}
\end{theorem}
The proof of Theorem \ref{thm:gdofBC} is presented in Appendix \ref{sec:proofBC}. 

\section{Proof of Theorem \ref{thm:gdof}: Achievability} \label{sec:proof}
As noted previously, Regime $3$ in Theorem \ref{thm:gdof} is  trivial and Regime $4$ already follows from \cite{Chan_Geng_Jafar_secureBC}. Thus we only need the proof for Regimes 1 and 2. In this section we provide the proof of achievability which is quite straightforward.

	For Regimes 1 and 2 it suffices to find schemes for the respective corner points and complete the regions by time-sharing. The tuple $(d_1,d_2)=(\alpha, 0)$ is one of the corner points for both cases, and is trivial. For Regime 1 it remains to find an achievable scheme for the other corner point, $(\alpha-\beta, 1)$. This is easily seen by modifying the scheme of Figure \ref{fig:gdof1}(b), such that Transmitter $1$   sends his desired message only  in the top $\alpha-\beta$ levels, i.e., and only a jamming signal (Gaussian noise) below that. Thus the noise floor at Receiver $1$ is elevated to strength $\beta$, i.e., as high as the interfering signal, which guarantees security. Meanwhile, we let Transmitter 2 transmit at full power. This creates a point-to-point channel for Transmitter 1 where the desired link to Receiver $1$ has $\alpha-\beta$ GDoF, and creates a wiretap channel for Transmitter $2$ where the desired link to Receiver 2 has 1 GDoF and the eavesdropper link to Receiver $1$ has 0 GDoF. Employing a Gaussian codebook in the first point-to-point channel and a wiretap codebook in the second, we achieve $\alpha-\beta$ SGDoF for User 1 and $1$ SGDoF for User 2.
	
	For Regime $2$ the other corner point is $(0, 1-\alpha+\beta)$. This is also easily achieved by modifying  the scheme of Figure \ref{fig:gdof2}(b), such that Transmitter $1$   sends only a jamming signal (Gaussian noise) with its full power. This raises the noise floor at Receiver $1$ to power level $\alpha$. As in Figure \ref{fig:gdof2}(b), we reduce the transmit power at Transmitter $2$ so that the top $\beta-\alpha$ levels are empty, i.e., instead of the unit power constraint, Transmitter $2$ only transmits with power $P^{-(\beta-\alpha)}$. 
	This creates a wiretap channel for Transmitter 2, where the desired link to Receiver 2 has $1+\alpha-\beta$ GDoF, and the eavesdropper link to Receiver 1 has 0 GDoF.
	A wiretap codebook achieves $1+\alpha-\beta$ SGDoF for User 2 and 0 for User 1.

	\section{Proof of Theorem \ref{thm:gdof}: Converse}\label{sec:converse}
	The  single user bound, $d_2\leq 1$, in Regime $1$ is trivial.
	Before presenting the proof of the weighted sum bounds, as preliminary background we need to introduce some definitions, sum-set inequalities, and a deterministic model, all of which originate in prior works on Aligned Images bounds.
	
	\subsection{Preliminaries from Prior Work} \label{sec:prelim}
	The following definitions are inherited from  \cite{Arash_Jafar, Arash_Jafar_sumset}.
	\subsubsection{Definitions}\label{sec:def}
	\begin{definition}[Power levels]			
		For $\lambda, P > 0$, define $\bP^\lambda \triangleq \lrfloor{\sqrt{P}^\lambda}$, and a set $\mathcal{X}_{\lambda}$ as 
		\begin{align}
			\mathcal{X}_{\lambda} = \left\{ 0, 1, 2, \cdots, \bP^\lambda -1 \right \},
		\end{align}	
		We refer to $P$ as \emph{power}, and $\lambda$ as \emph{power level} of $X\in\mathcal{X}_\lambda$. 
		For simplicity, we denote $\bP^1 = \bP$.
	\end{definition}
	\begin{definition} \label{def:subsection}
		For non-negative real numbers $X$, $\lambda_1$ and $\lambda_2$, where $\lambda_2 \geq \lambda_1 \geq 0$, we define a \emph{sub-section} of $X$ corresponding to \emph{interval} $(\lambda_1, \lambda_2)$, $(X)_{\lambda_1}^{\lambda_2}$, as
		\begin{align}
		(X)_{\lambda_1}^{\lambda_2} &\triangleq \lrfloor{ \frac{X - \bP^{\lambda_2} \lrfloor{\frac{X}{\bP^{\lambda_2}}}}{\bP^{\lambda_1}} }. \label{def:subsection-1}
		\end{align}
		We say that the  $(X)_{\lambda_1}^{\lambda_2}$  is a section of $X$ that sits at level $\lambda_1$, denoted as  $\ell \left( (X)_{\lambda_1}^{\lambda_2} \right) = \lambda_1 $, and has \emph{height} $\lambda_2-\lambda_1$, denoted as $\mathcal{T}\left( (X)_{\lambda_1}^{\lambda_2} \right) = \lambda_2 - \lambda_1$.
		Sub-sections $(X)_{\lambda_1}^{\lambda_2}$ and $(X)_{\lambda_1'}^{\lambda_2'}$ of $X \in \mathcal{X}_\lambda$ are \emph{disjoint} if intervals $(\lambda_1, \lambda_2)$ and $(\lambda_1', \lambda_2')$ are disjoint.  
	\end{definition}

		Figure \ref{fig:subsection} illustrates this partitioning of $X$ into various sub-sections. Similarly, for a set of non-negative real numbers $\boldsymbol{X} = \{X(t):~t\in[n]\}$, we define a sub-section $(\bX)_{\lambda_1}^{\lambda_2}$ as 
		\begin{align}
		(\bX)_{\lambda_1}^{\lambda_2} &\triangleq \{ (X(t))_{\lambda_1}^{\lambda_2}: ~ t \in [n]\}. \label{def:subsection-2}
		\end{align}
		Note that the same partitioning is applied to every element in the set.
		Levels and heights are similarly defined; i.e., $\ell\left ( (\bX)_{\lambda_1}^{\lambda_2} \right) = \lambda_1$, and $\mathcal{T}\left( (\bX)_{\lambda_1}^{\lambda_2} \right) = \lambda_2 - \lambda_1$.
		Sub-section sets $(\bX)_{\lambda_1}^{\lambda_2}$ and $(\bX)_{\lambda_1'}^{\lambda_2'}$ are disjoint if intervals  $(\lambda_1, \lambda_2)$ and $(\lambda_1', \lambda_2')$ are disjoint. 

	For $X \in \mathcal{X}_\lambda$ and $\lambda \geq \lambda_2 \geq \lambda_1 \geq 0$, sub-section $(X)^{\lambda_2}_{\lambda_1}$ can be loosely interpreted in terms of the $\bP$-ary expansion of $X$. 
	The $\bP$-ary expansion of $X$ is represented as $X = x_\lambda x_{\lambda-1} \cdots x_2 x_1$, which is equivalent to a string of length $\lambda$ in which each symbol $x_i \in \{0, 1, \cdots, \bP-1\}$.
	In this sense, what $(X)_{\lambda_1}^{\lambda_2}$ retrieves from $X$ is a sub-string $x_{\lambda_2} x_{\lambda_2-1} \cdots x_{\lambda_1+1}$ in the middle of $X$. 
	A case that appears frequently in this work is $\lambda_2 = \lambda$ and $\lambda_1 = \lambda-\mu$. 
	The corresponding sub-section $(X)^{\lambda}_{\lambda-\mu}$, denoted as $(X)^\mu$ and referred to as top-$\mu$ sub-section of $X$, retrieves from $X$ the leftmost length-$\mu$ sub-string $x_\lambda x_{\lambda-1} \cdots x_{\lambda-\mu+1}$ comprised of the first $\mu$ most significant symbols in $X.$ Similar to (\ref{def:subsection-2}), for a set of non-negative real numbers $\bX$ with each element in $\mathcal{X}_\lambda$, we define $(\bX)^{\mu} = \{(X)^\mu:~ X \in \bX  \}$.

		While this interpretation is helpful, the coarse understanding is an  oversimplification, as indeed all $\lambda, \lambda_1$ and $\lambda_2$ can take arbitrary non-negative real values. 
	Such partitioning is essentially a generalization of the original symbol partitioning with binary representations that appeared in the ADT model in  \cite{Avestimehr_Diggavi_Tse}. The generalization is needed because of our focus on finite precision CSIT.

	\begin{figure}[t]
		\centering
		
		\begin{tikzpicture}[scale=1.5]
			\def \d {0.5}
			\def \w {1}
			\def \h {2}
			
			\draw [line width = 1] (0,0) -- ({{6*\d+4*\w}},0);
			
			\draw [dashed, line width = 0.5] (0, \h) -- ({2.75*\d+2*\w}, \h);
			\draw [dashed, line width = 0.5] ({\d+\w}, {0.6*\h}) -- ({2.75*\d+2*\w}, {0.6*\h});
			\draw [dashed, line width = 0.5] ({\d+\w}, {0.7*\h}) -- ({3.75*\d+3*\w}, {0.7*\h});
			\draw [dashed, line width = 0.5] ({\d+\w}, {0.2*\h}) -- ({3*\d+2*\w}, {0.2*\h});
			\draw [dashed, line width = 0.5] ({\d+\w}, {0.4*\h}) -- ({5*\d+4*\w}, {0.4*\h});
			
			\draw [<->, line width = 0.5] ({\d/2}, 0) -- ({\d/2}, \h) node [pos = 0.5, left] {$\lambda$};
			\draw [<->, line width = 0.5] ({2.5*\d+2*\w}, {0.6*\h}) -- ({2.5*\d+2*\w}, {\h}) node [pos = 0.5, right] {$\mu$};
			\draw [<->, line width = 0.5] ({3*\d+2.5*\w}, {0*\h}) -- ({3*\d+2.5*\w}, {0.2*\h}) node [pos = 0.5, left] {$\lambda_1$};
			\draw [<->, line width = 0.5] ({3.5*\d+3*\w}, {0*\h}) -- ({3.5*\d+3*\w}, {0.7*\h}) node [pos = 0.75, right] {$\lambda_2$};
			\draw [<->, line width = 0.5] ({4.5*\d + 4*\w}, {0*\h}) -- ({4.5*\d + 4*\w}, {0.4*\h}) node [pos = 0.5, right] {$\lambda_3$};
			
			\draw [fill=red!60] (\d, 0) rectangle ({\d+\w}, \h) node [pos = 0.5] {$X$};
			
			\draw [fill = red!10] ({2*\d+\w}, {0.6*\h}) rectangle ({2*\d + 2*\w}, \h) node [pos = 0.5] {$A_1$};
			\draw [fill = red!30] ({3*\d+2*\w}, {0.2*\h}) rectangle ({3*\d + 3*\w}, {0.7*\h}) node [pos = 0.5] {$A_2$};
			\draw [fill = red!50] ({4*\d+3*\w}, {0*\h}) rectangle ({4*\d + 4*\w}, {0.4*\h}) node [pos = 0.5] {$A_3$};
			
		\end{tikzpicture}

		\caption{\it \small An illustration of Definition \ref{def:subsection}. Sub-section $A_1 = (X)^\lambda_{\lambda-\mu}$ has level $\ell(A_1) = \lambda-\mu$ and height $\mathcal{T}(A_1) = \mu$. Sub-section $A_2 = (X)^{\lambda_2}_{\lambda_1}$ has level $\ell(A_2) = \lambda_1$ and height $\mathcal{T}(A_2) = \lambda_2 - \lambda_1$. Sub-section $A_3 = (X)^{\lambda_3}_{0}$ has level $\ell(A_3) = 0$ and height $\mathcal{T}(A_3) = \lambda_3$. Note that $A_1$ and $A_3$ are disjoint when $\lambda - \mu \geq \lambda_3$.}
		\label{fig:subsection}
	\end{figure}
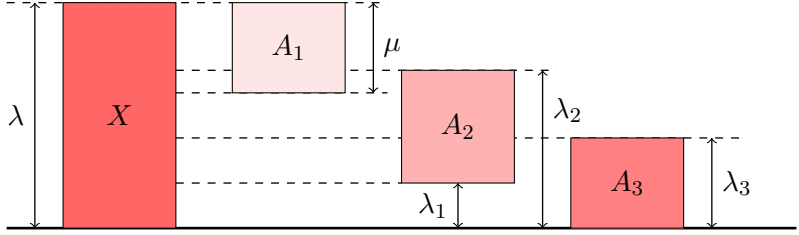
	
	\begin{definition}[Bounded density assumption] \label{def:bounded}
		We define  $\mathcal{G}$ as a set of real-valued random variables that satisfies the following conditions (collectively referred to as the bounded density assumption),
		\begin{enumerate}
			\item The magnitudes of all random variables in $\mathcal{G}$ are bounded away from infinity and zero; i.e., there exists a constant $\Delta>1$ such that $|g| \in \left(\frac{1}{\Delta}, \Delta \right)$ for all $g \in \mathcal{G}$.
			\item There exists a finite constant $f_\text{max} > 0$, such that for all finite disjoint subsets $\mathcal{G}_1$, $\mathcal{G}_2$ of $\mathcal{G}$, the joint probability density function of the random variables in $\mathcal{G}_1$, conditioned on the random variables in $\mathcal{G}_2$, exists and is bounded above by $f_\text{max}^{|\mathcal{G}_1|}$.
		\end{enumerate}
	\end{definition}

	\begin{definition}[Finite-precision linear combination]
		For $X_1 \in \mathcal{X}_{\eta_1}$ and $X_2 \in \mathcal{X}_{\eta_2}$, define $X_1 \boxplus_\chg X_2$ as 
		\begin{align}
		X_1 \boxplus_\chg X_2 & \triangleq \lrfloor{G_{1}X_1} + \lrfloor{G_{2}X_2},
		\end{align}
		where $G_{i}$ are distinct random variables in $\chg$ satisfying the bounded density assumption.
		For two sets of random variables of the same cardinality, $\bX_1 = \{X_1(t) \in \mathcal{X}_{\eta_1}:~t\in[n]\}$ and $\bX_2  = \{X_2(t)\in \mathcal{X}_{\eta_2}:~t\in[n]\}$ , we define $\bX_1 \boxplus_\chg \bX_2$ as 
		\begin{align}
		\bX_1 \boxplus_\chg \bX_2 &\triangleq \left \{ \lrfloor{G_{1}(t)X_1(t)} + \lrfloor{G_{2}(t)X_2(t)} :~ t\in[n] \right \},
		\end{align}
		where $G_{i}(t)$ are distinct random variables in $\chg$ satisfying the bounded density assumption.
		The subscript $\chg$ of operator $\boxplus$ may be omitted if no ambiguity arises.
	\end{definition}
	\subsubsection{Key Sumset Inequalities}\label{sec:sumset}
	Our proof leans heavily on the sum-set inequalities based on Aligned Image sets from \cite[Theorem 4]{Arash_Jafar_sumset}. 
	While \cite{Arash_Jafar_sumset} presents these sum-set inequalities in generalized forms, the following simplified forms of those inequalities, taken from \cite[Lemma 1]{Wang_Jafar_LimitedCoop}, will be useful for our purpose.

	\begin{lemma}\label{lemma:ais-sumset}
		Let $\mu, \nu >0$, $T(t) \in \mathcal{X}_\mu$, $U(t) \in \mathcal{X}_\nu$ for $t \in [n]$, and $\bT = \{T(t): t \in [n]\},\bU = \{U(t): t \in [n]\}$.
		Let $S_T$ and $S_U$ be sets of finitely many disjoint sub-sections respectively of $\bT$ and $\bU$, and let $\{\bA_1, \bA_2, \cdots, \bA_M\}$ be a subset of $S_T \cup S_U$. 
		Let $\bV = \bT \boxplus_\chg \bU$.
		Then
		\begin{align}
			\Hg\left( \bV \middle | \mathcal{W} \right) &\geq \Hg\left( \bA_1, \bA_2, \cdots, \bA_M \middle | \mathcal{W} \right) + \nologP, \label{lemma:ais-sumset-1}
		\end{align}
		where $\mathcal{W}$ is a set of random variables satisfying $ I(\mathcal{W}, \bT, \bU; \chg) = 0$, and the following constraints on the levels and heights of $\bA_i$ hold for $i = 2, 3, \cdots, M$:
		\begin{align}
			\ell( \bA_i ) &\geq \mathcal{T}( \bA_1) + \mathcal{T}( \bA_2) + \cdots + \mathcal{T}( \bA_{i-1}). \label{lemma:ais-sumset-2}
		\end{align}
	\end{lemma}
	Constraint (\ref{lemma:ais-sumset-2}) in Lemma \ref{lemma:ais-sumset} has the following box-stacking interpretation.
	Let's consider the $t^\text{th}$ channel use only and drop the index for simplicity.
	We can imagine these random variable sub-sections as boxes with labels $A_1, A_2, \cdots, A_M$; box $A_i$ has height $\mathcal{T}(A_i)$ and originally sits on level $\ell(A_i)$ in either $T$ or $U$. 
	Then we stack the boxes in the index order of $A_1, A_2, \cdots, A_M$ from the ground. 
	Now in this stack  box $A_i$ sits above boxes $A_1, A_2, \cdots, A_{i-1}$, therefore it sits at level $\tilde{\ell}(A_i) = \mathcal{T}(A_1) + \mathcal{T}(A_2) + \cdots + \mathcal{T}(A_{i-1})$. 
	Constraint (\ref{lemma:ais-sumset-2}) says that the new level $\tilde{\ell}(A_i)$ cannot be higher than the level at which box $A_i$ originally sits in $T$ or $U$, which is $\ell(A_i)$. 
	In other words, constraint (\ref{lemma:ais-sumset-2}) is satisfied if, during retrieving these boxes in $T$ or $U$ and stacking them up from ground, there is no need to elevate  any  of them above their original level. 
	Note that while constraints (\ref{lemma:ais-sumset-2}) seem to fix the stacking order according to the indices of the sub-sections, on the right-hand-side of (\ref{lemma:ais-sumset-1}) the entropy of the sub-sections does not depend on the index ordering. 
	So one can arbitrarily rearrange the indices of the sub-sections and test the constraints in (\ref{lemma:ais-sumset-2}) with the the permuted ordering.
	In other words, if there exists a stacking order of these boxes with no need to lift up any of them during stacking, then the sum-set inequality (\ref{lemma:ais-sumset-1}) holds.
	Figure \ref{fig:stack} and \ref{fig:sumset} illustrate some ways to stack the boxes (sub-sections) which satisfy or violate constraints (\ref{lemma:ais-sumset-2}).
	
	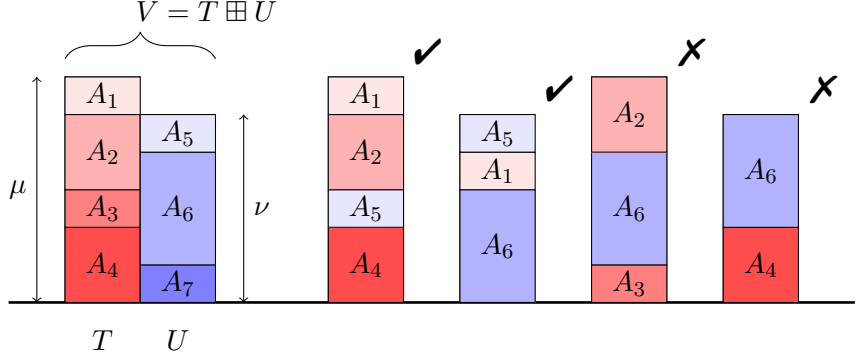
\begin{figure}[t]
		\centering
		\begin{tikzpicture}[scale = 1]
			\def \d {0.75}
			\def \w {1}
			\def \l {0.5}
			
			\draw [line width = 1] (0,0) -- ({7*\d + 6*\w}, 0);
			
			\coordinate (T) at ({\d}, 0);
			\coordinate (U) at ({\d+\w}, 0);
			\coordinate (1) at ({3*\d+2*\w}, 0);
			\coordinate (2) at ({4*\d+3*\w}, 0);
			\coordinate (3) at ({5*\d+4*\w}, 0);
			\coordinate (4) at ({6*\d+5*\w}, 0);
			
			\draw (T) rectangle ($(T) + (\w, {6 *\l})$);
			\draw [fill = red!70] (T) rectangle ($(T) + (\w, {2 *\l})$) node [pos = 0.5] {$A_4$};
			\draw [fill = red!50]  ($(T) + (0, {2 *\l})$) rectangle ($(T) + (\w, {3 *\l})$) node [pos = 0.5] {$A_3$};
			\draw [fill = red!30]  ($(T) + (0, {3 *\l})$) rectangle ($(T) + (\w, {5 *\l})$) node [pos = 0.5] {$A_2$};
			\draw [fill = red!10]  ($(T) + (0, {5 *\l})$) rectangle ($(T) + (\w, {6 *\l})$) node [pos = 0.5] {$A_1$};
			\draw (U) rectangle ($(U) + (\w, {5 *\l})$);
			\draw [fill = blue!50] (U) rectangle ($(U) + (\w, {1 *\l})$) node [pos = 0.5] {$A_7$};
			\draw [fill = blue!30]  ($(U) + (0, {\l})$) rectangle ($(U) + (\w, {4*\l})$) node [pos = 0.5] {$A_6$};
			\draw [fill = blue!10]  ($(U) + (0, {4*\l})$) rectangle ($(U) + (\w, {5 *\l})$) node [pos = 0.5] {$A_5$};
			
			\draw (1) rectangle ($(1) + (\w, {6 *\l})$) node [above right] {\CheckmarkBold};
			\draw [fill = red!70]  ($(1) + (0, {0 *\l})$) rectangle ($(1) + (\w, {2 *\l})$) node [pos = 0.5] {$A_4$};
			\draw [fill = blue!10]  ($(1) + (0, {2 *\l})$) rectangle ($(1) + (\w, {3 *\l})$) node [pos = 0.5] {$A_5$};
			\draw [fill = red!30]  ($(1) + (0, {3 *\l})$) rectangle ($(1) + (\w, {5 *\l})$) node [pos = 0.5] {$A_2$};
			\draw [fill = red!10]  ($(1) + (0, {5 *\l})$) rectangle ($(1) + (\w, {6 *\l})$) node [pos = 0.5] {$A_1$};
			
			\draw (2) rectangle ($(2) + (\w, {5 *\l})$) node [above right] {\CheckmarkBold};
			\draw [fill = blue!30]  ($(2) + (0, {0 *\l})$) rectangle ($(2) + (\w, {3 *\l})$) node [pos = 0.5] {$A_6$};
			\draw [fill = red!10]  ($(2) + (0, {3 *\l})$) rectangle ($(2) + (\w, {4 *\l})$) node [pos = 0.5] {$A_1$};
			\draw [fill = blue!10]  ($(2) + (0, {4 *\l})$) rectangle ($(2) + (\w, {5 *\l})$) node [pos = 0.5] {$A_5$};
			
			\draw (3) rectangle ($(3) + (\w, {6 *\l})$) node [above right] {\XSolidBrush};
			\draw [fill = red!50]  ($(3) + (0, {0 *\l})$) rectangle ($(3) + (\w, {1 *\l})$) node [pos = 0.5] {$A_3$};
			\draw [fill = blue!30]  ($(3) + (0, {1 *\l})$) rectangle ($(3) + (\w, {4 *\l})$) node [pos = 0.5] {$A_6$};
			\draw [fill = red!30]  ($(3) + (0, {4 *\l})$) rectangle ($(3) + (\w, {6 *\l})$) node [pos = 0.5] {$A_2$};
			
			\draw (4) rectangle ($(4) + (\w, {5 *\l})$) node [above right] {\XSolidBrush};	
			\draw [fill = red!70]  ($(4) + (0, {0 *\l})$) rectangle ($(4) + (\w, {2 *\l})$) node [pos = 0.5] {$A_4$};
			\draw [fill = blue!30]  ($(4) + (0, {2 *\l})$) rectangle ($(4) + (\w, {5 *\l})$) node [pos = 0.5] {$A_6$};
			
			\node at ($(T) + ({\w/2}, -0.5)$) {$T$};
			\node at ($(U) + ({\w/2}, -0.5)$) {$U$};
			
			\draw [decorate,decoration={brace,amplitude=10pt ,raise=4pt},yshift=0pt]
			($(T) + (0, {6*\l + 0.1})$) -- ($(U) + (\w, {6*\l + 0.1})$) node [black,midway, xshift = 25pt, yshift = 22pt] {$V = T \boxplus U$};
			
			\draw [<->] ({\d/2}, 0) -- ({\d/2}, {6*\l}) node [pos = 0.5, left] {$\mu$};
			\draw [<->] ({2*\w + 1.5*\d}, 0) -- ({2*\w + 1.5*\d}, {5*\l}) node [pos = 0.5, right] {$\nu$};
			
		\end{tikzpicture}
		\caption{\it \small An illustration of the box-stacking interpretation of Lemma \ref{lemma:ais-sumset}. 
			The bounds $\Hg(V|\mathcal{W}) \geq \Hg(A_1, A_2, A_4, A_5|\mathcal{W})$ and $\Hg(V|\mathcal{W}) \geq \Hg(A_1, A_5, A_6|\mathcal{W})$ are implied by Lemma \ref{lemma:ais-sumset} in the GDoF sense because the boxes appearing in these inequalities can be stacked without elevating any of them above their original levels in $T$ or $U$, as illustrated in the  two stacks marked with a \CheckmarkBold. 		On the other hand, Lemma \ref{lemma:ais-sumset} implies neither the bound $\Hg(V|\mathcal{W}) \geq \Hg(A_2, A_3, A_6|\mathcal{W})$ nor $\Hg(V|\mathcal{W}) \geq \Hg(A_4, A_6|\mathcal{W})$, because there is no way to stack the boxes appearing in these inequalities without elevating some of them above their original level in $T$ or $U$, as shown in the two stacks marked with \XSolidBrush. }
		\label{fig:stack}
	\end{figure}

	\subsubsection{Deterministic Model} \label{sec:det}
	To facilitate the use of Aligned Images bounds, we define a deterministic model as in \cite{Arash_Jafar}. In this deterministic model, the inputs are
	\begin{align}
	A(t) &= \lrfloor{\bP^{\alpha} X_1(t) } \mod \bP^{\alpha},\\
	B(t) &= \lrfloor{\bP^{\max\{1, \beta \}} X_2(t)} \mod \bP^{\max\{ 1, \beta \}},
	\end{align}
	and the outputs are
	\begin{align}
	\overline{Y}_1(t) &= \lrfloor{G_{11}(t) A(t)} + \lrfloor{ G_{12}(t)  \bP^{-(1-\beta)^+} B(t)  } , \label{eq:modeldet1}\\
	\overline{Y}_2(t) &= \lrfloor{G_{22}(t) \bP^{ -(\beta-1)^+ } B(t) }.\label{eq:modeldet2}
	\end{align}
	Note that $A(t) \in \mathcal{X}_{\alpha}$ and $B(t) \in \mathcal{X}_{\max\{1, \beta\}}$.
	Let $\bA = \{ A(t):~ t \in [n] \}$, and $\bB = \{ B(t):~ t \in [n] \}$, and $\bbY_i = \{ \overline{Y}_i(t):~ t \in [n] \}$ for $i = 1,2$.
	It can be shown that the GDoF of the Gaussian model are bounded above by the GDoF of the deterministic model, accounting for both decoding and secrecy constraints, as described by the following lemma.
	
	\begin{lemma} \label{lemma:noloss}
		\begin{align}
		\Ig(W_i; \bY_i) &\leq \Ig(W_i; \bbY_i) + no(\log P) && \forall i = 1,2, \label{eq:noloss-1}\\
		\Ig(W_j; \bbY_i) &\leq \Ig(W_j; \bY_i) + no(\log P) && \forall i, j = 1,2, i\neq j.\label{eq:noloss-2} 
		\end{align}
	\end{lemma}
	
	The proof of Lemma \ref{lemma:noloss} is identical to that of Lemma 5.1 in \cite{Chan_Geng_Jafar_secureBC}.

	\subsection{Useful Lemmas} \label{sec:lemmas}
	
	With the preliminaries in place, we now proceed to the task of proving the converse for Theorem \ref{thm:gdof}, starting with the following lemmas. The first lemma is a straightforward consequence of the secrecy constraint (\ref{eq:noloss-2}).
	
	\begin{lemma} \label{lemma:secrecy}
		Let $\umu = (\beta -\alpha)^+$ and $\lmu = (\alpha -\beta)^+$. Then we have,
		\begin{align}
		\Ig(W_2; \bbY_1, W_1) &= \nologP, \label{lemma:secrecy-1}\\
		\Ig( \bbY_1; W_2 | W_1, \bAmu, \bBmu ) &= \nologP, \label{lemma:secrecy-2} \\
		\Ig( W_2; W_1, \bAmu, \bBmu ) &= \nologP. \label{lemma:secrecy-3}
		\end{align}
	\end{lemma}
	\begin{proof}
		\begin{align}
		\Ig(W_2; \bbY_1, W_1) &= \Ig(W_2; \bbY_1) + \Ig(W_2; W_1 | \bbY_1) \label{eq:secrecy-1}\\
		&\leq  \Ig(W_2; \bbY_1) + \Hg(W_1 | \bbY_1) \label{eq:secrecy-2}\\
		&\leq \Ig(W_2; \bY_1) + \Hg(W_1 | \bY_1) + \nologP  \label{eq:secrecy-3}\\
		&= \nologP. \label{eq:secrecy-4}
		\end{align}
		We apply the chain rule to get (\ref{eq:secrecy-1}), and the definition of mutual information to obtain (\ref{eq:secrecy-2}). 
		Next, we obtain (\ref{eq:secrecy-3}) by applying (\ref{eq:noloss-1}) and (\ref{eq:noloss-2}). 
		Finally, we apply the secrecy constraint (\ref{eq:secrecy}) and Fano's inequality 
		to obtain (\ref{eq:secrecy-4}).
		
		To show equality (\ref{lemma:secrecy-2}) and (\ref{lemma:secrecy-3}), we note that from $\bbY_1$ one can obtain $\bAmu$ and $\bBmu$, and then apply the chain rule; more specifically,
		\begin{align}
			\nologP &= \Ig(W_2; \bbY_1, W_1)\\ 
			&= \Ig(W_2; \bbY_1, W_1, \bAmu, \bBmu)\\ 
			&=  \Ig(W_2; W_1, \bAmu, \bBmu) + \Ig(W_2; \bbY_1| W_1, \bAmu, \bBmu).
		\end{align}
		Equality (\ref{lemma:secrecy-2}) and (\ref{lemma:secrecy-3}) thus hold as mutual information is non-negative.
	\end{proof}

	The following lemma bounds from above the entropy difference, in the GDoF sense, of finite-precision linear combinations of random variables in terms of their power levels. 
	It is adapted from Lemma 1 of \cite{Arash_Jafar_mimoSymIC} and hence its proof is omitted.
	\begin{lemma}\label{lemma:ais}
		Let $\mu = \max_{i=1,2}\{ \mu_i\}$ and $\nu = \max_{i=1,2}\{ \nu_i\}$, where $\mu_i, \nu_i >0, i = 1,2$. 
		Let $T(t)\in \mathcal{X}_{\nu}$ and $U(t) \in \mathcal{X}_{\mu}$ for $t \in [n]$; $\bT = \{T(t):~t\in[n]\}$ and $\bU = \{ U(t):~ t \in [n] \}$. 
		Let $\bV_i =  (\bT)^{\mu_i} \boxplus_{\chg_i} (\bU)^{\nu_i}$, where $i=1,2$, and $\chg = \chg_1 \cup \chg_2$ is a set of random variables satisfying the bounded density assumption.
		Then
		\begin{align}
			\Hg(\bV_1 | \mathcal{W}) - \Hg(\bV_2 | \mathcal{W}) \leq \max \{ \mu_1 - \mu_2,  \nu_1 - \nu_2 \}^+ \log P + \nologP,
		\end{align}
		where $\mathcal{W}$ is a set of random variables satisfying $I(\mathcal{W}, \bT, \bU; \chg) = 0$. 
	\end{lemma}
	
	An important issue that arises in applications of Aligned Images bounds is  that of translating between `linear combinations of sub-sections' on one hand, and `sub-sections of linear combinations' on the other. Sum-set inequalities are formulated in \cite{Arash_Jafar_sumset} in terms of linear combinations of various sub-sections of input signals, but converse arguments often involve sub-sections of \emph{output} signals, i.e., sub-sections of linear combinations of input signals. Understanding the extent to which these two notions can be related remains an open problem in general \cite{Chan_Jafar_3to1}. For our present purpose, however, because we only need the `top' sub-sections, such a relationship is obtained in the following lemma.
	
	\begin{lemma}\label{lemma:const}
		Let $\lambda, \mu, \nu$ be real numbers satisfying $\lambda \geq \mu > 0$ and $\nu \geq 0$.
		Let $T \in \mathcal{X}_{\nu + \lambda}$ and $U\in \mathcal{X}_{\nu + \mu}$.
		Then
		\begin{align}
			\Hg( ( T \boxplus U )^\lambda ) = \Hg( (T)^\lambda \boxplus (U)^\mu  ) + O(1),
		\end{align}
		where $\chg$ is a set of random variables satisfying the bounded density assumption. 
	\end{lemma}
	The proof of Lemma \ref{lemma:const} is relegated to Appendix \ref{sec:const}.


	The next lemma provides an important lower bound on the entropy of a finite-precision linear combination of random variables based on Lemma \ref{lemma:ais-sumset} and the submodularity of entropy.
	
	\begin{lemma} \label{lemma:sumset}
		Let $P, \mu, \nu \geq 0$, and let $p, q > 0$ satisfy $\frac{1}{2} \leq \frac{p}{q} \leq 1$ and $\frac{p}{q} \in \mathbb{Q}$.
		Let $T(t) \in \mathcal{X}_{q+\mu}$ and $U(t)\in \mathcal{X}_{q+\nu}$ for $t \in [n]$; $\bT = \{T(t):~t\in[n] \}$ and $\bU = \{ U(t): t \in [n]\}$.
		Let $\bV = \bT \boxplus_\chg \bU$, where $\chg$ is a set of random variables satisfying the bounded density assumption. 
		Then
		\begin{align}
		2p\Hg( \bV | \mathcal{W}, (\bT)^\mu, (\bU)^\nu ) \geq q \Hg( (\bT)^{p+\mu}, (\bU)^{p+\nu} | \mathcal{W}, (\bT)^\mu, (\bU)^\nu ) + \nologP, \label{eq:sumset-result}
		\end{align}
		where $\mathcal{W}$ is a set of random variables satisfying $I(\mathcal{W}, \bT, \bU; \chg) = 0$. 
	\end{lemma}
	\begin{proof}
		Since $\frac{p}{q} \in \mathbb{Q}$, there exists $\ell \in \mathbb{R}$ and $\tilde{p}$, $\tilde{q} \in \mathbb{N}$, such that $ p = \tilde{p} \ell$ and $q = \tilde{q} \ell$. For all $t \in [n]$, define sub-sections of $T(t)$ and $U(t)$ as
		\begin{align} \label{eq:sumset-A}
		A_i(t) = \begin{cases}
			(T(t))^{q-(i-1)\ell}_{q - i \ell} & \text{if } 1 \leq i \leq \tilde{p} \\
			(U(t))^{q-(i-\tilde{p}-1)\ell}_{q - (i-\tilde{p}) \ell} & \text{if }\tilde{p}+1 \leq i \leq 2\tilde{p}\\
		\end{cases},
		\end{align}
		and $\bA_i = \{ A_i(t):~t \in[n]\}$ for $ i \in [2 \tilde{p}]$.
		Then by Lemma \ref{lemma:ais-sumset}, for $ i \in [2 \tilde{p}]$ the following holds:
		\begin{align}
			\Hg(\bV | \mathcal{W}, (\bT)^\mu, (\bU)^\nu) 
			\geq \Hg( \bA_i, \bA_{i+1}, \cdots, \bA_{i+q-1}  | \mathcal{W}, (\bT)^\mu, (\bU)^\nu) + \nologP, \label{eq:sumset-0}
		\end{align}
		where we implicitly use modulo-$2\tilde{p}$ arithmetic in the indices; e.g., $i_0 = i_{2\tilde{p}}$.
		Lemma \ref{lemma:ais-sumset} is applied in the following way.
		After removing top-$\mu$ sub-section of $\bT$ and top-$\nu$ sub-section of $\bU$, we take the top-$p$ sub-section of the remaining $\bT$ and $\bU$, and evenly slice them into $\tilde{p}$ boxes, each of which has height $\ell$.
		The boxes in $\bT$ are then indexed from top to bottom with 1 to $\tilde{p}$, and those in $\bU$ are indexed likewise with $\tilde{p}+1$ to $2\tilde{p}$. 
		Conditioned on the top-$\mu$ sub-section of $\bT$ and the top-$\nu$ sub-section of $\bU$, Lemma \ref{lemma:ais-sumset} implies that the entropy of $\bT \boxplus_\chg \bU$ is no less than the joint entropy of the boxes whose indices are within a circular sliding window of size $\tilde{q}$.
		This can be verified with the box-stacking interpretation of Lemma \ref{lemma:ais-sumset}.
		See Figure \ref{fig:sumset} for an illustration of the procedure above.
		
		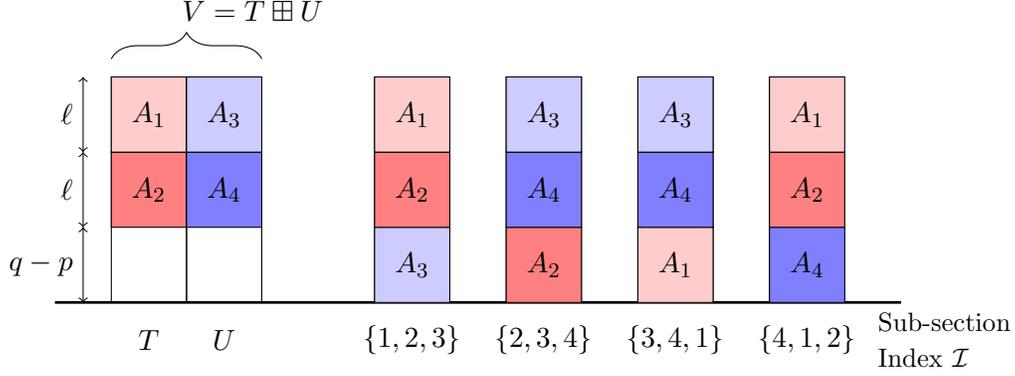
\begin{figure}[t]
			\centering
			\begin{tikzpicture}[scale = 1]
				\def \d {0.75}
				\def \w {1}
				\def \l {1}
				
				\draw [line width = 1] (0,0) -- ({7*\d + 6*\w}, 0);
				
				\coordinate (T) at ({\d}, 0);
				\coordinate (U) at ({\d+\w}, 0);
				\coordinate (1) at ({3*\d+2*\w}, 0);
				\coordinate (2) at ({4*\d+3*\w}, 0);
				\coordinate (3) at ({5*\d+4*\w}, 0);
				\coordinate (4) at ({6*\d+5*\w}, 0);
				
				\draw (T) rectangle ($(T) + (\w, {3 *\l})$);
				\draw [fill = red!50]  ($(T) + (0, {1 *\l})$) rectangle ($(T) + (\w, {2*\l})$) node [pos = 0.5] {$A_2$};
				\draw [fill = red!20]  ($(T) + (0, {2 *\l})$) rectangle ($(T) + (\w, {3 *\l})$) node [pos = 0.5] {$A_1$};
				\draw (U) rectangle ($(U) + (\w, {3 *\l})$);
				\draw [fill = blue!50]  ($(U) + (0, {1*\l})$) rectangle ($(U) + (\w, {2*\l})$) node [pos = 0.5] {$A_4$};
				\draw [fill = blue!20]  ($(U) + (0, {2*\l})$) rectangle ($(U) + (\w, {3 *\l})$) node [pos = 0.5] {$A_3$};
				
				\draw (1) rectangle ($(1) + (\w, {3 *\l})$);
				\draw [fill = blue!20]  ($(1) + (0, {0 *\l})$) rectangle ($(1) + (\w, {1 *\l})$) node [pos = 0.5] {$A_3$};
				\draw [fill = red!50]  ($(1) + (0, {1 *\l})$) rectangle ($(1) + (\w, {2 *\l})$) node [pos = 0.5] {$A_2$};
				\draw [fill = red!20]  ($(1) + (0, {2 *\l})$) rectangle ($(1) + (\w, {3 *\l})$) node [pos = 0.5] {$A_1$};
				
				\draw (2) rectangle ($(2) + (\w, {3 *\l})$);
				\draw [fill = red!50]  ($(2) + (0, {0 *\l})$) rectangle ($(2) + (\w, {1 *\l})$) node [pos = 0.5] {$A_2$};
				\draw [fill = blue!50]  ($(2) + (0, {1 *\l})$) rectangle ($(2) + (\w, {2 *\l})$) node [pos = 0.5] {$A_4$};
				\draw [fill = blue!20]  ($(2) + (0, {2 *\l})$) rectangle ($(2) + (\w, {3 *\l})$) node [pos = 0.5] {$A_3$};
				
				\draw (3) rectangle ($(3) + (\w, {3 *\l})$);
				\draw [fill = red!20]  ($(3) + (0, {0 *\l})$) rectangle ($(3) + (\w, {1 *\l})$) node [pos = 0.5] {$A_1$};
				\draw [fill = blue!50]  ($(3) + (0, {1 *\l})$) rectangle ($(3) + (\w, {2 *\l})$) node [pos = 0.5] {$A_4$};
				\draw [fill = blue!20]  ($(3) + (0, {2 *\l})$) rectangle ($(3) + (\w, {3 *\l})$) node [pos = 0.5] {$A_3$};

				\draw (4) rectangle ($(4) + (\w, {3 *\l})$);
				\draw [fill = blue!50]  ($(4) + (0, {0 *\l})$) rectangle ($(4) + (\w, {1 *\l})$) node [pos = 0.5] {$A_4$};
				\draw [fill = red!50]  ($(4) + (0, {1 *\l})$) rectangle ($(4) + (\w, {2 *\l})$) node [pos = 0.5] {$A_2$};
				\draw [fill = red!20]  ($(4) + (0, {2 *\l})$) rectangle ($(4) + (\w, {3 *\l})$) node [pos = 0.5] {$A_1$};

				\node at ($(T) + ({\w/2}, -0.5)$) {$T$};
				\node at ($(U) + ({\w/2}, -0.5)$) {$U$};
				\node at ($(1) + ({\w/2}, -0.5)$) {$\{1,2,3\}$};
				\node at ($(2) + ({\w/2}, -0.5)$) {$\{2,3,4\}$};
				\node at ($(3) + ({\w/2}, -0.5)$) {$\{3,4,1\}$};
				\node at ($(4) + ({\w/2}, -0.5)$) {$\{4,1,2\}$};
				\node [text width = 60, align = left] at ($(4) + ({\w+2*\d}, -0.5)$) {\small Sub-section Index $\mathcal{I}$};
				
				\draw [<->] ({\d/2}, 0) -- ({\d/2}, {\l}) node [pos = 0.5, left] {$q-p$};
				\draw [<->] ({\d/2}, {\l}) -- ({\d/2}, {2*\l}) node [pos = 0.5, left] {$\ell$};
				\draw [<->] ({\d/2}, {2*\l}) -- ({\d/2}, {3*\l}) node [pos = 0.5, left] {$\ell$};
				
				\draw [decorate,decoration={brace,amplitude=10pt ,raise=4pt},yshift=0pt]
				($(T) + (0, {3*\l + 0.1})$) -- ($(U) + (\w, {3*\l + 0.1})$) node [black,midway, xshift = 25pt, yshift = 22pt] {$V = T \boxplus U$};
				
			\end{tikzpicture}	

			\caption{\it \small 
				An illustration of how the sum-set inequality in Lemma \ref{lemma:ais-sumset} is applied to the proof of Lemma \ref{lemma:sumset}.
				In this case, $\mu = \nu = 0$, $p = 2$, and $q=3$, which implies that $\ell = 1$, $\tilde{p} = 2$ and $\tilde{q}=3$.
				The left most consecutive bars shows $\bV = \bT \boxplus_\chg \bU$ and some sub-sections of $\bT$ and $\bU$ taken by (\ref{eq:sumset-A}).
				The right four bars list all possible sub-section index sets obtained by a circular sliding window of size $q =3$.
				Seeing that all sub-sections in each index set satisfy the box-stacking interpretation (All boxes can be stacked without elevating any above their original levels), Lemma \ref{lemma:ais-sumset} implies that $\Hg(\bV | \mathcal{W}) \geq \Hg(\bA_\mathcal{I} | \mathcal{W})$, where $\mathcal{I}$ is one of the sub-section index sets, and $\bA_\mathcal{I} = \{A_i:~i \in \mathcal{I}\}$. 
				Summing up these inequalities and applying the submodularity of entropy, one can obtain (\ref{eq:sumset-result}).}
			\label{fig:sumset}
		\end{figure}
		
		Adding up (\ref{eq:sumset-0}) for all $i \in [2\tilde{p}]$, we have		
		\begin{align}
		& 2p \Hg( \bV | \mathcal{W}, (\bT)^\mu, (\bU)^\nu ) \notag \\
		& = \ell 2 \tilde{p} \Hg( \bV | \mathcal{W}, (\bT)^\mu, (\bU)^\nu ) \label{eq:sumset-1}\\
		& \geq \ell \sum_{i=1}^{2\tilde{p}}  \Hg( \bA_i, \bA_{i+1}, \cdots, \bA_{i+\tilde{q}-1}  | \mathcal{W}, (\bT)^\mu, (\bU)^\nu ) + \nologP \label{eq:sumset-2}\\
		& \geq \ell \tilde{q} \Hg( \bA_1, \bA_2, \cdots, \bA_{2\tilde{p}}  | \mathcal{W}, (\bT)^\mu, (\bU)^\nu) + \nologP \label{eq:sumset-3}\\
		& \geq q \Hg( (\bT)^{p+\mu}, (\bU)^{p+\nu} | \mathcal{W}, (\bT)^\mu, (\bU)^\nu ) + \nologP. \label{eq:sumset-4}
		\end{align}
		Step (\ref{eq:sumset-1}) holds since $p = \tilde{p}\ell$.
		Step (\ref{eq:sumset-3}) follows from the sub-modularity
		\footnote{
		Let $\{X_1, X_2 \cdots, X_n\}$ be a set of random variables, then for $1 \leq k \leq n$, the submodularity of entropy implies:
		\begin{align}
		\sum_{i=1}^{n} H(X_i, X_{i+1}, \cdots, X_{i+k-1}) \geq k H(X_1, X_2, \cdots, X_n),
		\end{align}
		where modulo-$n$ arithmetic is implicitly used in the inidices, e.g., $i_0 = i_n$. 
		} of entropy, and (\ref{eq:sumset-4}) holds because $q = \tilde{q}\ell$, and one can recover $(\bT)^{\mu+p}$ and $(\bU)^{\nu + p}$ from $ \{\bA_i:i \in [2\tilde{p}]\},  (\bT)^\mu, (\bU)^\nu$, and  $\chg$ within bounded distortion. 
	\end{proof}

	\subsection{The Weighted-Sum Bounds in Regime 1 and 2} \label{sec:wsumBound12}
		We break down the proof into the following three lemmas. 
		Throughout this section, we define $\mu = \beta - \alpha, \umu = (\mu)^+$, $\lmu = (-\mu)^+$, and $\groupW = \{ W_1, \bAmu, \bBmu \}$.
		Note that in both Regime 1 and 2, we have $\umu \leq 1$.

\begin{lemma} \label{lemma:bottom}
	For $\lambda \geq 1 - \mu$ and $\umu \leq 1$, we have
	\begin{align}
		\Hg((\bbY_1)^\lambda | \groupW) \geq nR_2 + \Hg( (\bbY_1)^{\lambda - (1-\umu)} | \groupW ) + \nologP.
	\end{align}
\end{lemma}
\begin{proof}
	\begin{align}
		\Hg( (\bbY_1)^\lambda | \groupW) 
		&= \Hg( (\bA)^{\lambda - \umu} \boxplus (\bB)^{\lambda-\lmu} | \groupW) + \nologP \label{eq:bot-1}\\
		&\geq \Hg( (\bA)^{\lambda - 1} \boxplus (\bB)^{\lambda-\lmu} | \groupW ) + \nologP \label{eq:bot-2}\\
		&= \Hg(W_2| \groupW) + \Hg( (\bA)^{\lambda - 1} \boxplus (\bB)^{\lambda-\lmu} ) |\groupW, W_2 ) \notag \\
		&\qquad - \Hg(W_2 | \groupW, (\bA)^{\lambda - 1} \boxplus (\bB)^{\lambda-\lmu}   ) + \nologP \label{eq:bot-3}\\
		&= H(W_2) +  \Hg( (\bA)^{\lambda - 1} \boxplus (\bB)^{\lambda-\lmu} | \groupW, W_2) + \nologP \label{eq:bot-4}\\
		&\geq nR_2 + \Hg( (\bA)^{\lambda - 1} \boxplus (\bB)^{\lambda-1 + \mu}| \groupW, W_2 ) + \nologP \label{eq:bot-5}\\
		&= nR_2 + \Hg( (\bbY_1)^{\lambda - (1-\umu)} | \groupW, W_2) + \nologP \label{eq:bot-6}\\
		&= nR_2 + \Hg( (\bbY_1)^{\lambda - (1-\umu)} | \groupW ) + \nologP. \label{eq:bot-7}
	\end{align}
	First, equality (\ref{eq:bot-1}) holds because by Lemma \ref{lemma:const} one can recover $(\bA)^{\lambda - \umu} \boxplus (\bB)^{\lambda-\lmu}$  from $ (\bbY_1)^\lambda$ within bounded distortion.
	Then we apply Lemma \ref{lemma:ais} to obtain (\ref{eq:bot-2}), and apply the chain rule to obtain (\ref{eq:bot-3}).
	Equality (\ref{eq:bot-4}) holds for the following reasons: (a) equality (\ref{lemma:secrecy-3}) implies the first entropy term; (b) the last entropy term is of $\nologP$ is because, from $\bAmu$  and $(\bA)^{\lambda - 1} \boxplus (\bB)^{\lambda-\lmu}$, by Lemma \ref{lemma:const} one can recover $(\bB)^1$ within bounded distortion, which one can decode for $W_2$.
	Then we apply $nR_2 = H(W_2)$ and Lemma \ref{lemma:ais} to obtain (\ref{eq:bot-5}). Note that Lemma \ref{lemma:ais} is applicable because in Regimes 1 and 2, $1-\mu = 1+\alpha-\beta \geq (\alpha-\beta)^+ = \lmu$.
	Equality (\ref{eq:bot-6}) holds because by Lemma \ref{lemma:const}, $ (\bbY_1)^{\lambda - (1-\umu)}$ can be recovered from $(\bA)^{\lambda - 1} \boxplus (\bB)^{\lambda-1+\mu}$ within bounded distortion.
	Finally, we arrive at (\ref{eq:bot-7}) due to (\ref{lemma:secrecy-2}).
	\end{proof}

	In the next lemma, we show that the part of codeword $\bA$ corresponding to the same power levels as the part of $\bB$ carrying $W_2$ has entropy no less than $H(W_2) = nR_2$.
	Intuitively, this must be so because $W_2$ needs to be hidden from Receiver 1, and for this the `jamming signal' must be at least as big as $W_2$.

	\begin{lemma} \label{lemma:hide}
		\begin{align}
			\Hg( (\bA)^{1-\mu} | \groupW, (\bB)^1 ) \geq nR_2 + \nologP.
		\end{align}
	\end{lemma}
\begin{proof}
	\begin{align}
		\Hg( (\bB)^1 | \groupW ) 
			&\leq \Hg( (\bbY_1)^{1+\lmu} | \groupW ) + \nologP \label{eq:hide-1}\\
			&= \Hg(  (\bbY_1)^{1+\lmu} | \groupW, W_2 ) + \nologP \label{eq:hide-2}\\
			&\leq \Hg( (\bA)^{1-\mu}, (\bB)^1 |  \groupW, W_2) + \nologP \label{eq:hide-3}\\
			&= \Hg( (\bB)^1 |  \groupW, W_2 ) + \Hg( (\bA)^{1-\mu} |  \groupW, W_2, (\bB)^1 )+ \nologP \label{eq:hide-4} \\
			&\leq \Hg( (\bB)^1 |  \groupW, W_2 ) + \Hg( (\bA)^{1-\mu} | \groupW, (\bB)^1 )+ \nologP.	 \label{eq:hide-5}
	\end{align}
	First, we apply Lemma \ref{lemma:ais} to obtain inequality (\ref{eq:hide-1}). Note that $(\bbY_1)^{1+\lmu}$ is well-defined because $\beta > 1$ in Regime 1 and 2, and implies that $\max\{\alpha, \beta\} > 1 + \lmu$.
	Equality (\ref{eq:hide-2}) holds due to (\ref{lemma:secrecy-2}).
	Inequality (\ref{eq:hide-3}) is true because $\mu = \beta -\alpha < 1$ in Regime 1 and 2, and  $(\bbY_1)^{1+\lmu}$ can be recovered by Lemma \ref{lemma:const} within bounded distortion from $(\bA)^{1-\mu} \boxplus (\bB)^1 $, which is a function of $(\bA)^{1-\mu}$ and $(\bB)^1$.
	Then we apply the chain rule to obtain (\ref{eq:hide-4}), and apply the fact that conditioning reduces entropy to obtain (\ref{eq:hide-5}).	
	
	By swapping terms in (\ref{eq:hide-5}), we have
	\begin{align}
		 \Hg( (\bA)^{1-\mu} |  \groupW, (\bB)^1 ) 
		 	&\geq \Hg( (\bB)^1 |\groupW ) - \Hg( (\bB)^1 |  \groupW, W_2 ) + \nologP \label{eq:hide-6}\\
		 	&= \Ig( (\bB)^1; W_2 | \groupW) + \nologP \label{eq:hide-7}\\
		 	&= \Ig( (\bB)^1, \groupW; W_2 ) - I(\groupW; W_2) + \nologP \label{eq:hide-8}\\
		 	&= \Ig( (\bB)^1, \groupW; W_2 ) + \nologP \label{eq:hide-9}\\
		 	&\geq \Ig( (\bB)^1; W_2) + \nologP \label{eq:hide-10}\\
		 	&\geq \Ig(\bbY_2; W_2) + \nologP \label{eq:hide-11}\\
		 	&= nR_2 + \nologP. \label{eq:hide-12}
	\end{align}
	We apply the definition of mutual information to obtain (\ref{eq:hide-7}), the chain rule to obtain (\ref{eq:hide-8}), and (\ref{lemma:secrecy-3}) to obtain (\ref{eq:hide-9}).
	Then we remove $\groupW$ to obtain (\ref{eq:hide-10}).
	Finally, we apply data processing inequality to obtain (\ref{eq:hide-11}), and  Fano's inequality to obtain (\ref{eq:hide-12}).
\end{proof}

The third lemma is a lower bound for the entropy $\Hg( \bbY_1 | \groupW )$.

\begin{lemma} \label{lemma:rational}
	For $\umu \leq 1$, we have
	\begin{align}
		\Hg(\bbY_1 | \groupW) &\geq \frac{\min\{ \beta, \alpha \}}{1-\umu} nR_2 + \nologP.
	\end{align}
\end{lemma}
\begin{proof}
	Let $\min\{ \beta, \alpha \}  = k ( 1 - \umu ) + \gamma$, where $k$ is a non-negative integer, and $\gamma$ satisfies either $\gamma = 0$ or $1-\umu < \gamma < 2(1-\umu)$
	\footnote{The existence of such $k$ and $\gamma$ can be shown as follows. In Regime 1, since $\beta > 1$, we can find $k , \gamma$, where either $\gamma = 0$ or $1 < \gamma <2$, such that $\beta = k+\gamma$. On the other hand, in Regime 2, since $\alpha > 1+\alpha-\beta$, we can find $k ,\gamma$ such that $\alpha = k (1+\alpha-\beta) + \gamma$ with either $\gamma =0$ or $ 1+\alpha-\beta < \gamma < 2(1+\alpha-\beta).$}.
	As an intermediate result, we claim that 
	\begin{align}
		\Hg( (\bbY_1)^{\gamma + |\mu|} | \groupW) \geq \frac{\gamma}{1-\umu} nR_2 + \nologP. \label{eq:rat-gamma1}
	\end{align}
	The inequality is trivial when $\gamma = 0$. 
	If $\gamma \neq 0$, we can find a non-decreasing sequence $\{r_i\}$ with $r_i \in \mathbb{Q}$ and $\lim_{i \rightarrow \infty}r_i = \gamma$, and a non-increasing sequence $\{m_i\}$ with $m_i \in \mathbb{Q}$ and $\lim_{i \rightarrow \infty}m_i = 1-\umu$.
	\footnote{
	Such a non-increasing sequence $\{m_i\}$ and a non-increasing sequence $\{r_i\}$ can be constructed by the decimal representation of $1-\umu$ and $\gamma$, respectively. For example, let $0.\mu_1\mu_2\cdots \mu_i$ be the $i-$decimal of $1-\umu$, where $\mu_j \in \{0,1,\cdots, 9\}$ for $j \in [i]$. We may let $m_i = 0.\mu_1\mu_2\cdots \mu_i + 10^{-i} = \left( \lrfloor{(1-\umu)\times 10^i}  + 1\right) \times 10^{-i}$, which is a rational number no less than $1-\umu$. On the other hand, let $0.\gamma_1\gamma_2\cdots \gamma_i$ be the $i-$decimal of $\gamma$, where $\gamma_j \in \{0,1,\cdots, 9\}$ for $j \in [i]$. We may let $r_i = 0.\gamma_1\gamma_2\cdots \gamma_i = \lrfloor{\gamma \times 10^i} \times 10^{-i}$, which is a rational number no greater than $\gamma$.
	}
	Let $N = \min\left\{ i \middle | \frac{m_i}{r_i} < 1 \right\}$. Such $N$ exists, because as $i \rightarrow \infty$, we have $r_i \rightarrow \gamma$, $m_i \rightarrow 1 - \umu$, and $ \frac{1}{2} < \frac{1-\umu}{\gamma} < 1$.

	For $i \geq N$, we have
	\begin{align}
		&\Hg( (\bbY_1)^{\gamma + |\mu|} | \groupW) \\
			&\geq \Hg((\bbY_1)^{r_i + |\mu|} | \groupW) + \nologP \label{eq:rat-1} \\
			&= \frac{1}{2m_i} \left( 2m_i \Hg((\bbY_1)^{r_i + |\mu|} | \groupW )  \right)+ \nologP \label{eq:rat-2}\\
			&\geq \frac{r_i}{2m_i} \Hg(  (\bA)^{m_i + \lmu}, (\bB)^{m_i + \umu} | \groupW  ) + \nologP \label{eq:rat-3}\\
			&\geq \frac{r_i}{2m_i} \Hg(  (\bA)^{1-\mu}, (\bB)^{1} | \groupW  ) + \nologP \label{eq:rat-4}\\
			&= \frac{r_i}{2m_i} \left(  \Hg( (\bB)^{1} | \groupW ) + \Hg(  (\bA)^{1-\mu} | \groupW, (\bB)^{1}  ) \right) + \nologP  \label{eq:rat-5}\\
			&\geq \frac{r_i}{2m_i}	\left(  \Hg(W_2| \groupW) + \Hg( (\bB)^{1} | \groupW, W_2  ) - \Hg(W_2 | \groupW, (\bB)^{1} ) + nR_2 \right)	+ \nologP \label{eq:rat-6}\\
			&\geq \frac{r_i}{m_i} nR_2 + \nologP. \label{eq:rat-7}
	\end{align}
	Inequality (\ref{eq:rat-1}) holds because of Lemma \ref{lemma:ais} and the fact that $r_i \leq \gamma$.
	Then we multiply and divide the entropy term by $2m_i$ to get (\ref{eq:rat-2}), and apply\footnote{To apply Lemma \ref{lemma:sumset}, we define $\bT = (\bA)^{r_i + \lmu} \in \mathcal{X}_{r_i + \lmu}, \bU = (\bB)^{r_i + \umu} \in \mathcal{X}_{r_i +\umu}, p = m_i,$ and $q = r_i$. This leads to $\bV = (\bA)^{r_i + \lmu}\boxplus (\bB)^{r_i + \umu}$, which by Lemma \ref{lemma:const} can be recovered from $(\bbY_1)^{r_i + |\mu|}$ within bounded distortion.} 
	Lemma \ref{lemma:sumset} to obtain (\ref{eq:rat-3}).
	Inequality (\ref{eq:rat-4}) holds because of Lemma \ref{lemma:ais} and the fact that $m_i + \lmu \geq 1 - \umu + \lmu = 1-\mu$, and $m_i + \umu \geq 1$.
	Next we apply the chain rule to get (\ref{eq:rat-5}), and apply the chain rule and Lemma \ref{lemma:hide} to get (\ref{eq:rat-6}).
	Equality (\ref{eq:rat-7}) follows from (\ref{eq:rat-6}) due to the following reasons: (a) we apply (\ref{lemma:secrecy-3}) and $nR_2 = H(W_2)$ to the first entropy term; (b) the second entropy term is non-negative; and (c) $W_2$ can be decoded from $(\bB)^1$, which makes the third entropy term  $\nologP$.
	Since inequality (\ref{eq:rat-7}) is valid for all $i \geq N$, we have
	\begin{align}
		\Hg( (\bbY_1)^{\gamma + |\mu|} | \groupW ) & \geq \lim_{i \rightarrow \infty} \frac{r_i}{m_i} nR_2 + \nologP = \frac{\gamma}{1-\umu} nR_2 + \nologP. \label{eq:rat-gamma2}
	\end{align}
	
	Next, based on the intermediate result (\ref{eq:rat-gamma1}), we show the following lower bound.
	\begin{align}
		\Hg(\bbY_1 | \groupW) \geq knR_2 + \Hg( (\bbY_1)^{|\mu| + \gamma} | \groupW )  + \nologP. \label{eq:botrat-0} 
	\end{align} 
	This bound is reduced to (\ref{eq:rat-gamma1}) when $k=0$ because of the following identity
	\begin{align}
		\max\{\beta, \alpha\} = |\mu| + \min \{\beta, \alpha\} = |\mu| + k(1-\umu) + \gamma. \label{eq:mugamma}
	\end{align}
	On the other hand, when $k \geq 1$, we apply Lemma \ref{lemma:bottom} as follows.
	\begin{align}
		\Hg(\bbY_1 | \groupW) 
		& \geq nR_2 + \Hg( (\bbY_1)^{\max\{\beta, \alpha\} - (1-\umu)} | \groupW ) + \nologP \label{eq:botrat-1} \\
		& \geq 2nR_2 + \Hg( (\bbY_1)^{\max\{\beta, \alpha\} - 2(1-\umu)} | \groupW )  + \nologP \label{eq:botrat-2} \\
		& \geq \cdots \notag\\
		& \geq knR_2 + \Hg( (\bbY_1)^{\max\{\beta, \alpha\} - k(1-\umu)} | \groupW )  + \nologP \label{eq:botrat-3} \\
		& = knR_2 + \Hg( (\bbY_1)^{|\mu| + \gamma} | \groupW )  + \nologP. \label{eq:botrat-4} 
	\end{align}
	Lemma \ref{lemma:bottom} can be applied to (\ref{eq:botrat-1}) -- (\ref{eq:botrat-3}) because in both Regime 1 and 2, we have $\umu \leq 1$  and $\max\{\alpha, \beta\} - (k-1) (1-\umu) \geq 1-\mu$
	\footnote{This can be seen by the following. First by (\ref{eq:mugamma}) we have $\max\{\alpha, \beta\} - (k-1) (1-\umu) = 1-\umu + \gamma + |\mu| = 1 + \gamma + \lmu$. In Regime 1, we have $1 + \gamma + \lmu \geq 1+\lmu = 1-\mu$, while in Regime 2, we have $1 + \gamma + \lmu \geq 1 \geq 1-\mu$. }.	
	Next we apply (\ref{eq:mugamma}) to obtain (\ref{eq:botrat-4}).
	
	Finally, we plug (\ref{eq:rat-gamma1}) into (\ref{eq:botrat-0}), and get 
	\begin{align}
		\Hg(\bbY_1 | \groupW)
		& \geq knR_2 + \frac{\gamma}{1 -\umu} nR_2  + \nologP \label{eq:botrat-5}\\
		& = \frac{\min\{ \beta, \alpha \}}{1 - \umu} nR_2  + \nologP, \label{eq:botrat-6}
	\end{align}
	where equality (\ref{eq:botrat-6}) holds by applying the identity $\min \{\beta, \alpha\}  = k(1-\umu) + \gamma$.
\end{proof}

To finish the proof of the weighted-sum bound, we start by applying Fano's inequality as follows.
\begin{align}
	nR_1 & \leq \Ig(\bbY_1; W_1) + \nologP \label{eq:fano-1}\\
		 & \leq \Ig(\bbY_1, (\bB)^{\umu}; W_1) + \nologP \label{eq:fano-2} \\
		 & = \Ig(\bbY_1; W_1 | (\bB)^{\umu}) + \Ig((\bB)^{\umu}; W_1) + \nologP \label{eq:fano-3}\\
		 & = \Ig( \bbY_1; W_1 | (\bB)^{\umu} ) + \Ig((\bbY_2)^{\umu}; W_1) + \nologP \label{eq:fano-4}\\
		 & = \Ig( \bbY_1; W_1 | (\bB)^{\umu} ) + \nologP \label{eq:fano-5}\\
		 & \leq \Hg(\bbY_1 | (\bB)^{\umu} ) - \Hg( \bbY_1 | \groupW ) + \nologP \label{eq:fano-6}\\
		 & \leq \alpha n \log \bP - \frac{\min\{ \beta, \alpha \}}{1- \umu} nR_2+ \nologP \label{eq:fano-7}
\end{align}
Inequality (\ref{eq:fano-2}) holds because adding $(\bB)^{\umu}$ does not hurt the mutual information.
Then we apply the chain rule to get (\ref{eq:fano-3}).
Since ${\umu} \leq 1$ in Regime 1 and 2, $(\bB)^{\umu}$ can be converted into $(\bbY_2)^{\umu}$ within bounded distortion by Lemma \ref{lemma:const}, and as a result we have (\ref{eq:fano-4}).
Equality (\ref{eq:fano-5}) holds due to (\ref{eq:noloss-2}) and the secrecy constraint (\ref{eq:secrecy}), and the fact that $\umu \leq 1$.
Then seeing that $\{(\bB)^{\umu}\} \subset \groupW$, inequality (\ref{eq:fano-6}) is obtained by applying the fact that conditioning reduces entropy.
To obtain inequality (\ref{eq:fano-7}), first we apply the uniform bound to the first entropy in (\ref{eq:fano-6}) as follows:
\begin{align}
	\Hg(\bbY_1 | (\bB)^{\umu}) \leq \Hg( (\bbY_1)^\alpha_0 ) \leq \alpha n \log P + \nologP.
\end{align}  
And then we apply Lemma \ref{lemma:rational} to the second entropy in (\ref{eq:fano-6}) . Note that Lemma \ref{lemma:rational} is applicable since $\umu \leq 1$ in Regime 1 and 2.

Finally by applying the definition of GDoF, we get
\begin{align}
		& d_1 + \frac{\min\{ \beta, \alpha \}}{1- \umu} d_2 = \lim_{P \rightarrow \infty} \frac{R_1 + \frac{\min\{ \beta, \alpha \}}{1- \umu} R_2}{\frac{1}{2}\log P}\leq \alpha \label{eq:fano-8}\\
	\implies & \begin{cases}
		d_1 + \beta d_2 \leq \alpha & \text{if } \alpha >\beta\\
		\frac{d_1}{\alpha} + \frac{d_2}{1 + \alpha - \beta} \leq 1 & \text{if } \beta-1 < \alpha \leq \beta
	\end{cases}. \label{eq:fano-9}
\end{align}
Inequalities (\ref{eq:fano-9}) are the desired weighted-sum bounds for the respective Regime 1 and 2. 
Thus, we complete the proof.$\hfill\square$

\section{Conclusion} \label{sec:conclusion}
Motivated by robustness concerns that are paramount in secure communications, in this work we  study the robust GDoF of secure communication over a $2$ user $Z$ interference channel. The combination of security, robustness and GDoF optimality makes this problem uniquely challenging relative to prior work, while the $Z$ channel setting limits the number of parameters sufficiently to allow a  GDoF characterization for all parameter regimes. In the process we also explore the scope of  sum-set inequalities based on Aligned Images bounds that were recently introduced in \cite{Arash_Jafar_sumset}, which involve joint entropies of various sub-sections of input signals. We found that these new sum-set inequalities, combined with sub-modularity properties of entropy, are sufficient to characterize the robust secure GDoF region of a $Z$-interference channel (as well as a further generalization to the corresponding broadcast channel setting). The result shows that the GDoF benefits of structured jamming, e.g., aggregate decoding and cancellation of jammed signals, are entirely lost under finite precision CSIT. The result reaffirms the hypothesis that random codes may be enough for approximate capacity characterizations under robust assumptions. Thus, while the fundamental limits of structured codes under ideal assumptions remain both practically fragile and theoretically intractable,  there remains hope that continued advances in Aligned Images converse bounds may eventually place within reach a robust network information theory of wireless networks, based on the understanding of the fundamental limits of random codes.

\appendix
\section{Proof of Lemma \ref{lemma:gdof_perfect}} \label{sec:proofPerfect}

In this section we present the proof of the SGDoF region $\mathcal{D}_{\tiny {IC}}^{\tiny p}$. 
The converse bounds are available from Lemma 8 of \cite{Chen_Jam} (for single-user bound) and Lemma 2 of \cite{He_Yener_BoundGIC} (for the sum bound). 
The converse bounds are tight in Regime 3 and 4 defined in Theorem \ref{thm:gdof}, as $\mathcal{D}_{\tiny {IC}}^{\tiny p} = \mathcal{D}_{\tiny {IC}}^{\tiny f.p.}$ in these regimes, and the schemes for finite precision CSIT also apply to the case with perfect CSIT.
The remaining part to be shown is the achievability of $\mathcal{D}_{\tiny {IC}}^{\tiny p}$ in Regime 1 and 2.

In the following presentation of the schemes, without loss of generality we work on the simplified ZIC with all channel gains normalized to be $1$; i.e.,
\begin{align}
	Y_1(t) &= \sqrt{P^\alpha} X_1(t) + \sqrt{P^\beta}X_2(t) + Z_1(t), \label{eq:modelPerfect1}\\
	Y_2(t) &= \sqrt{P} X_2(t) + Z_2(t),\label{eq:modelPerfect2}
\end{align}
where $t \in [n]$, $Z_1(t), Z_2(t) \sim \mathcal{N}(0,1)$ and $X_1(t), X_2(t)$ are subject to unit input power constraint. 
This can be done by normalizing the inputs and the outputs of the original model (\ref{eq:model1}) and (\ref{eq:model2}) with the channel coefficients, which are known at both sides. Also we set the noise variances to  unity since they are inconsequential to the GDoF analysis.

\subsection{The Achievability in Regime 1} \label{sec:gdof_case1}

The corner points of $\mathcal{D}_{\tiny {IC}}^{\tiny p}$ in Regime 1 are $(d_1, d_2) = (\alpha, 1)$ and $(\beta-1, 1)$. The former is trivial, and time sharing achieves all tuples on the line segment between these two point. So we show the tuple $(\beta-1, 1)$ is achievable with a scheme based on lattice alignment and aggregate decoding.

Let $Q \triangleq \lrfloor{ \sP{\alpha-\epsilon} }$, $Q_J \triangleq \lrfloor{ \sP{\alpha-1-\epsilon} }$, and $A = 8 \sP{2\epsilon}$, where $\epsilon > 0$. 
In the following we suppress the channel-use index $t$ for brevity.
Define $X_1 = V_{11} + J_1 + V_{12}$ and $X_2 = V_2$, where $V_{11}, J_1, V_{12}, V_2$ are drawn respectively from the following sets (referred to as lattices):
\begin{align}
	V_{11} \in \Gamma_{11} &\triangleq A\sP{-\beta} \times \lrbr{ 0, \pm Q, \pm 2Q, \cdots, \pm \lrfloor{\sP{\beta - \alpha - \epsilon}} Q }, \\
	J_{1} \in \Gamma_{J} &\triangleq A \sP{-\beta} \times \lrbr{ 0, \pm Q_J, \pm 2Q_J, \cdots, \pm \lrpar{\lrfloor{\tfrac{1}{8} \sP{1 - \epsilon}  } -1 } Q_J }, \\
	V_{12} \in \Gamma_{12} &\triangleq A\sP{-\beta} \times \lrbr{ 0, \pm 1, \pm 2, \cdots, \pm \lrpar{\lrfloor{\tfrac{1}{4} \sP{ \alpha - 1 - 2\epsilon } } -1} }, \\
	V_{2} \in \Gamma_{2} &\triangleq A \sP{-\alpha} \times \lrbr{0, \pm Q_J, \pm 2Q_J, \cdots, \pm \lrpar{\lrfloor{\tfrac{1}{8} \sP{1-\epsilon} } - 1} Q_J}.
\end{align}
where for a real number $\xi$ and a finite set of integers $\{x_1, x_2, \cdots, x_n\}$, we define their product $\xi \times \{x_1, x_2, \cdots, x_n\} \triangleq \{\xi x_1, \xi x_2, \cdots, \xi x_n\}$. 
Note that such a choice of $A, Q, Q_J$, along with the lattices $\Gamma_{11}, \Gamma_{J}, \Gamma_{12}$ and $\Gamma_{2}$,  satisfies the unit input power constraint.

Let $V_{11}, V_{12}, J_1$ and  $V_2$ be independent and uniformly distributed in their respective lattices. 
Message $W_1$ is split into two parts, which are respectively encoded into $V_{11}$ and $V_{12}$, and message $W_2$ is encoded into $V_2$, all with wiretap codebooks. 
The following rates are achievable under secrecy constraints \cite[Theorem 4]{Xie_Ulukus}:
\begin{align}
	R_1 &\geq I(Y_1; V_{11}, V_{12}), \label{eq:case1R1}\\
	R_2 &\geq I(Y_2; V_2) - I(Y_1; V_2 | V_{11}, V_{12}). \label{eq:case1R2}
\end{align}

We follow the argument in \cite{Xie_Ulukus, Bresler_Tse} to bound these rates from below.
First we bound $I(Y_1; V_{11}, V_{12})$ from below as follows.
\begin{align}
	& I(Y_1; V_{11}, V_{12}) \\
	&= H(V_{11}, V_{12}) - H(V_{11}, V_{12} | Y_1) \label{eq:case1R1-1}\\
	&\geq \lrpar{  \log |\Gamma_{11}|+ \log |\Gamma_{12}|  } \lrpar{1 - \Pr [ \hat{V}_{11} \neq V_{11} \text{ or } \hat{V}_{12} \neq V_{12}]} - 1 \label{eq:case1R1-2}\\
	&= \lrpar{  \log \lrpar{ 2 \lrfloor{ \sP{\beta - \alpha -\epsilon } } + 1} + \log \lrpar{ 2\lrfloor{\tfrac{1}{4}  \sP{\alpha-1-2\epsilon}}-1 }  } \lrpar{1 - \Pr [ \hat{V}_{11} \neq V_{11} \text{ or } \hat{V}_{12} \neq V_{12}]} - 1 \label{eq:case1R1-3}\\
	&\geq \lrpar{ \beta - 1 - 3\epsilon } \log \bP \lrpar{1 - \Pr [ \hat{V_{11}} \neq V_{11} \text{ or } \hat{V_{12}} \neq V_{12}]} - 3. \label{eq:case1R1-4}
\end{align}
In (\ref{eq:case1R1-2}), $\hat{V}_{11}$ and $\hat{V}_{12}$ follow the nearest-neighbor decoding rule and are respectively defined as 
\begin{align}
	\hat{V}_{11} &\triangleq \arg \min_{V_{11} \in \Gamma_{11} } \lrabs{Y_1 - \sP{\beta} V_{11}} , \label{eq:case1V11}\\
	\hat{V}_{12} &\triangleq \arg \min_{V_{12} \in \Gamma_{12}} \lrabs{\tilde{Y}_1 - AQ_J \lrbkt{\frac{\tilde{Y}_1}{AQ_J}} - \sP{\beta} V_{12} }, \label{eq:case1V12}
\end{align}
where $\tilde{Y}_1 \triangleq Y_1 - \sP{\beta} \hat{V}_{11} $, and $[x]$ rounds $x$ to its nearest integer for all $x \in \mathbb{R}$.
Inequality (\ref{eq:case1R1-2}) holds due to Fano's inequality and the fact that $V_{i1}$ is uniformly taken from $\Gamma_{1i}$, where $i = 1,2$.
Inequality (\ref{eq:case1R1-4}) holds for $P$ large enough because for $x \geq 2$, we have
\begin{align}
	\log(2 \lrfloor{x} - 1) \geq \log x. \label{eq:loglbd}
\end{align}

Next we follow steps similar to (\ref{eq:case1R1-1}) -- (\ref{eq:case1R1-4}) to bound $I(Y_2; V_2)$ from below as follows
\begin{align}
	I(Y_2; V_2) 
	&= H(V_2) - H(V_2|Y_2) \label{eq:case1R2-11}\\
	& = \lrpar{\log |\Gamma_2|} \lrpar{ 1- \Pr [ \hat{V}_2 \neq V_2]} -1 \label{eq:case1R2-12}\\
	& = \log \lrpar{  2\lrfloor{ \tfrac{1}{8} \sP{1-\epsilon} } - 1 } \lrpar{ 1- \Pr [ \hat{V}_2 \neq V_2]} -1  \label{eq:case1R2-13}\\
	&\geq (1-\epsilon) \log \bP  \lrpar{ 1 - \Pr [ \hat{V}_2 \neq V_2] } - 4, \label{eq:case1R2-14}
\end{align}
where in (\ref{eq:case1R2-12}) $ \hat{V}_2$ is defined as 
\begin{align}
	\hat{V}_2 &\triangleq \arg \min_{V_2 \in \Gamma_2}  \lrabs{Y_2 - \sP{\alpha} V_2}, \label{eq:case1V2}
\end{align}

As for the negative term in (\ref{eq:case1R2}), $I(Y_1; V_2 | V_{11}, V_{12})$, it is bounded above as follows.
\begin{align}
	I(Y_1; V_2 | V_{11}, V_{12})
	& \leq I(Y_1; V_2 | V_{11}, V_{12}, Z_1) \label{eq:case1R2-21}\\
	&= I( \sP{\beta} J_1 + \sP{\alpha} V_2; V_2) \label{eq:case1R2-22}\\
	&= H( \sP{\beta} J_1 + \sP{\alpha} V_2) - H(\sP{\beta} J_1) \label{eq:case1R2-23}\\
	&\leq \log \lrpar{ 4\lrfloor{\tfrac{1}{8} \sP{1-\epsilon}} - 3} - \log \lrpar{  2 \lrfloor{ \tfrac{1}{8}\sP{1-\epsilon} } - 1 } \label{eq:case1R2-24} \\
	&\leq 1. \label{eq:case1R2-25}
\end{align}
Ineqaulity (\ref{eq:case1R2-21}) holds since $Z_1$ is independent of $V_2$, and (\ref{eq:case1R2-22}) follows because $(V_{11}, V_{12}, Z_1)$ is independent of $(J_1, V_2)$.
Inequality (\ref{eq:case1R2-24}) is true due to the uniform bound and the fact that $  \sP{\beta} J_1 + \sP{\alpha} V_2 $ takes value from the set $AQ_J \times \{0, \pm 1, \pm 2, \cdots, \pm 2 \lrpar{ \lrfloor{\tfrac{1}{8} \sP{1-\epsilon} }-1} \}$.
Finally (\ref{eq:case1R2-25}) holds when $P$ is large enough due to (\ref{eq:loglbd}).

It remains to find upper bounds of $\Pr [\hat{V}_{11} \neq V_{11} \text{ or } \hat{V}_{12} \neq \hat{V}_{12}]$ in (\ref{eq:case1R1-4}) and $\Pr[ \hat{V}_{2} \neq V_2]$ in (\ref{eq:case1R2-14}).	
They vanish as $P$ goes to infinity, as stated in the following lemma, whose proof is relegated to Appendix \ref{sec:case1prob}.

\begin{lemma}\label{lemma:case1prob}
	Given $\hat{V}_{11}, \hat{V}_{12},$ and $\hat{V}_{2}$ are respectively defined in (\ref{eq:case1V11}), (\ref{eq:case1V12}) and (\ref{eq:case1V2}), we have 
	\begin{align}
		&\lim_{P \rightarrow \infty} \Pr [\hat{V}_{11} \neq V_{11} \text{ or } \hat{V}_{12} \neq V_{12}] = 0, \\
		&\lim_{P \rightarrow \infty} \Pr[ \hat{V}_{2} \neq V_2] = 0.
	\end{align}
\end{lemma}

Finally, by respectively plugging (\ref{eq:case1R1-4}) into (\ref{eq:case1R1}), and plugging (\ref{eq:case1R2-14}) and (\ref{eq:case1R2-25}) into (\ref{eq:case1R2}), we get
\begin{align}
	R_1 &\geq \lrpar{ \beta - 1 - 3\epsilon} \tfrac{1}{2}\log P + o(\log \bP) = \lrpar{ \beta - 1 } \tfrac{1}{2}\log P + o(\log \bP)\\
	R_2 &\geq \lrpar{ 1 - \epsilon} \tfrac{1}{2}\log P + o(\log \bP) =  \tfrac{1}{2}\log P + o(\log \bP).
\end{align}
We arrive at $d_1 = \lim_{P \rightarrow \infty} \frac{R_1}{\frac{1}{2}\log P} = \beta-1$, and $d_2 = \lim_{P \rightarrow \infty} \frac{R_1}{\frac{1}{2}\log P} = 1$. 
Thus the secure GDoF tuple $(d_1, d_2) = (\beta-1, 1)$ is achievable with this scheme based on lattice alignment and aggregate decoding.

\subsection{The Achievability in Regime 2} \label{sec:gdof_case2}
The corner points of $\mathcal{D}_{\tiny {IC}}^{\tiny p}$ in Regime 1 are $(d_1, d_2) = (\alpha,0)$ and $(\beta-1, 1+\alpha-\beta)$. Following the same reason for the corner points of Regime 1, it remains to show $(\beta-1, 1+\alpha-\beta)$ is achievable, which is done with lattice alignment and aggregate decoding as well.

Let $Q \triangleq \lrfloor{ \sP{\alpha-1-\epsilon}}$ and $A \triangleq \sP{2\epsilon}$, where $\epsilon > 0$.
In the following we suppress the channel-use index $t$ for brevity. 
Define $X_1 = V_1 +J_1$ and $X_2 = V_2$, where
\begin{align}
	V_1 \in \Gamma_1 &\triangleq A \sP{-\beta} \times \lrbr{ 0, \pm 1, \pm 2, \cdots, \pm \lrpar{  \lrfloor{ \tfrac{1}{2} \sP{\alpha-1-2\epsilon}}- 1}  },\\
	J_1 \in \Gamma_J &\triangleq A \sP{-\beta} \times \lrbr{ 0, \pm Q, \pm 2Q, \cdots, \pm \lrfloor{ \sP{1-\alpha+\beta-\epsilon}}Q  },\\
	V_2 \in \Gamma_2 &\triangleq A \sP{-\alpha} \times \lrbr{ 0, \pm Q, \pm 2Q, \cdots, \pm \lrfloor{ \sP{1-\alpha+\beta-\epsilon}}Q  }.
\end{align}
Note that such a choice of $A, Q$ and the lattices $\Gamma_1, \Gamma_J$ and $\Gamma_2$ satisfies the unit input power constraint.

Let $V_1, J_1$ and $V_2$ be independent and uniformly distributed in their respective lattices. 
Message $W_1$ and $W_2$ are respectively encoded into $V_1$ and $V_2$ with wiretap codebooks of rate $R_1$ and $R_2$.
The following rates are achievable under the secrecy constraints \cite[Theorem 4]{Xie_Ulukus}:
\begin{align}
	R_1 &\geq I(Y_1; V_1), \label{eq:case2R1}\\
	R_2 &\geq I(Y_2; V_2) - I(Y_1; V_2 | V_1). \label{eq:case2R2}
\end{align}
To further bound these rates from below, we follow steps similar to (\ref{eq:case1R1-1}) -- (\ref{eq:case1R1-4}) in Appendix \ref{sec:gdof_case1}, and get a lower bound of $I(Y_1; V_1)$ as follows.
\begin{align}
	I(Y_1; V_1) &\geq \lrpar{ \alpha-1-2\epsilon} \tfrac{1}{2} \log P  \lrpar{1-\Pr[\hat{V}_1 \neq V_1]} -2. \label{eq:case2R1-4}
\end{align}
where $\hat{V}_1$ is defined as 
\begin{align}
	\hat{V}_1 = \arg \min_{V_1 \in \Gamma_1} \lrabs{ Y_1 - AQ \left[ \frac{Y_1}{AQ} \right] - \sP{\beta} V_1}. \label{eq:case2V1}
\end{align}
To get a lower bound of $I(Y_2; V_2)$ we follow steps identical to (\ref{eq:case1R2-11}) -- (\ref{eq:case1R2-14}) 
\begin{align}
	I(Y_2; V_2) &\geq \lrpar{1 - \alpha +\beta - \epsilon} \tfrac{1}{2} \log P  \lrpar{ 1 - \Pr[\hat{V}_2 \neq V_2]} - 2, \label{eq:case2R2-14}
\end{align}
where $\hat{V}_2$  is defined as 
\begin{align}
	\hat{V}_2 = \arg \min_{V_2 \in \Gamma_2} \lrabs{ Y_2 - \sP{} V_2}. \label{eq:case2V2}
\end{align}
And we can bound $I(Y_1; V_2 | V_1)$ from above by following steps similar to (\ref{eq:case1R2-21}) -- (\ref{eq:case1R2-25}) 
\begin{align}
	I(Y_1; V_2 | V_1) & \leq 1. \label{eq:case2R2-25}
\end{align}

With a similar reasoning to the one in Lemma \ref{lemma:case1prob}, one can show that for both $i = 1, 2$, $\Pr[\hat{V}_i \neq V_i] \rightarrow 0$ as $P\rightarrow \infty$.
Finally, by plugging (\ref{eq:case2R1-4}) into (\ref{eq:case2R1}), and by plugging (\ref{eq:case2R2-14}) and (\ref{eq:case2R2-25}) into (\ref{eq:case2R2}), we get
\begin{align}
	R_1 &\geq \lrpar{ \alpha-1-2\epsilon} \tfrac{1}{2}\log P + o(\log \bP) =  \lrpar{ \alpha-1}\tfrac{1}{2}\log P + o(\log \bP), \\
	R_2  &\geq \lrpar{1 - \alpha +\beta - \epsilon} \tfrac{1}{2}\log P  + o(\log \bP)  = \lrpar{1 - \alpha +\beta} \tfrac{1}{2}\log P + o(\log \bP).
\end{align}
By applying the definition of GDoF we get $d_1 = \lim_{P \rightarrow \infty} \frac{R_1}{\frac{1}{2}\log P} = \alpha-1$ and $d_2 = \lim_{P \rightarrow \infty} \frac{R_2}{\frac{1}{2}\log P} = 1- \alpha+\beta$.
Hence the GDoF tuple $(d_1, d_2) = (\alpha-1, 1-\alpha+\beta)$ is achievable with this scheme.

\subsection{Proof of Lemma \ref{lemma:case1prob} }\label{sec:case1prob}
Let event $\mathcal{E} \triangleq \lrbr{ Z_1 \middle | |Z_1| \geq \frac{A}{2}}$, and its complement denoted as $\mathcal{E}^c = \lrbr{ Z_1 \middle | |Z_1| < \frac{A}{2}}$.
Define $I_1 \triangleq \sP{\beta} V_{12} + Z_1$ and $I_2 \triangleq \sP{\beta} J_1 + \sP{\alpha} V_2 + I_1$.
Note that $Y_1 = \sP{\beta} V_{11} + I_2$ is the sum of a lattice point $\sP{\beta} V_{11}$ and an offset $I_2$. 
The lattice point is taken from the lattice $\sP{\beta}\times\Gamma_{11}$ with the minimum spacing $AQ$, while the offset, $I_2$, takes value from $\lrpar{ -\frac{AQ}{2}, \frac{AQ}{2}}$ when $\mathcal{E}^c$ happens.
So when $\mathcal{E}^c$ occurs, $V_{11}$ can be correctly decoded by (\ref{eq:case1V11}), and seeing that $Z_1 \sim \mathcal{N}(0,1)$, we have
\begin{align}
	\Pr[ \hat{V}_{11} \neq V_{11} ] \leq \Pr\{\mathcal{E}\} \leq 2 \exp\lrpar{ - \frac{1}{8} A^2}. \label{eq:case1pV11}
\end{align}

Next we  move on and argue that $V_{12}$ can be correctly decoded with (\ref{eq:case1V12}) when $V_{11}$ is correctly decoded and $\mathcal{E}^c$ occurs. 
Suppose $V_{11}$ is correctly decoded and removed from $Y_1$, resulting in the remaining $\tilde{Y}_1 = I_2 = \sP{\beta} J_1 + \sP{\alpha} V_2 + I_1$.
Note that $I_2$ is the sum of offset $I_1$ and a lattice point $\sP{\beta} J_1 + \sP{\alpha} V_2$, which is taken from lattice $\sP{\beta} \times \Gamma_J + \sP{\alpha} \times \Gamma_2$. \footnote{For two sets $\Gamma_1$ and $\Gamma_2$, define $\Gamma_1 + \Gamma_2 \triangleq \{a + b | a \in \Gamma_1, b \in \Gamma_2 \}$ as the sum set of $\Gamma_1$ and $\Gamma_2$.}
Such a lattice has the minimum spacing $AQ_J$. 
On the other hand, offset $I_1$ takes value from $\lrpar{ -\frac{AQ_J}{2}, \frac{AQ_J}{2}}$ when $\mathcal{E}^c$ happens.
As a result, when $\mathcal{E}^c$ occurs, $\tilde{Y}_1 - AQ_J \lrbkt{\frac{\tilde{Y}_1}{AQ_J}}  = I_1 = \sP{\beta} V_{12} + Z_1 $.
Note that, once again, $I_1$ is the sum a lattice point $\sP{\beta} V_{12}$, which is taken from lattice $\bP^\beta \times \Gamma_{12}$ with the minimum spacing $A$, and an offset $Z_1$, which is in $\lrpar{-\frac{A}{2}, \frac{A}{2}}$ if $\mathcal{E}^c$ happens. 
Therefore, $V_{12}$ can be correctly decoded by (\ref{eq:case1V12}) when $\mathcal{E}^c$ occurs and $V_{11}$ is correctly decoded, and 
\begin{align}
	Pr[ \hat{V}_{12} \neq V_{12} |\hat{V}_{11} = V_{11}  ] \leq \Pr\{\mathcal{E}\} \leq 2 \exp\lrpar{ - \frac{1}{8} A^2}. \label{eq:case1pV12}
\end{align}

Finally we can bound $\Pr[\hat{V}_{11} \neq V_{11} \text{ or } \hat{V}_{12} \neq V_{12}]$ as follows.
\begin{align}
	&\Pr[\hat{V}_{11} \neq V_{11} \text{ or } \hat{V}_{12} \neq V_{12}]\\
	&\leq \Pr[\hat{V}_{11} \neq V_{11}] + \Pr[\hat{V}_{12} \neq V_{12} ] \label{eq:case1p-1}\\
	&\leq \Pr[\hat{V}_{11} \neq V_{11}] + \Pr[\hat{V}_{12} \neq V_{12} | \hat{V}_{11} = V_{11}] \Pr[ \hat{V}_{11} = V_{11} ] \notag \\ 
	& \qquad + \Pr[\hat{V}_{12} \neq V_{12} | \hat{V}_{11} \neq V_{11}] \Pr[ \hat{V}_{11} \neq V_{11} ] \label{eq:case1p-2}\\
	& \leq  \Pr[\hat{V}_{11} \neq V_{11}] + \Pr[\hat{V}_{12} \neq V_{12} | \hat{V}_{11} = V_{11}] + \Pr[ \hat{V}_{11} \neq V_{11} ] \label{eq:case1p-3}\\
	&\leq 6 \exp\lrpar{-\frac{1}{8} A^2}, \label{eq:case1p-4}
\end{align}
where we use the union bound in (\ref{eq:case1p-1}), and the law of total probability in (\ref{eq:case1p-2}). 
Inequality (\ref{eq:case1p-4}) holds because of (\ref{eq:case1pV11})  and (\ref{eq:case1pV12}).
Since $A^2 = O(P^{2\epsilon})$ and $\epsilon > 0$, we have $\Pr[\hat{V}_{11} \neq V_{11} \text{ or } \hat{V}_{12} \neq V_{12}] \rightarrow 0$ as $P \rightarrow \infty$.

Note that $Y_2 = \sP{} V_2 + Z_2$ is the sum of a lattice point $\sP{} V_2$, which is taken from lattice $\sP{} \times \Gamma_2$ with the minimum spacing $A\sP{1-\alpha}Q_J$, and an offset $Z_2$, which is in $\lrpar{-\frac{A}{2}, \frac{A}{2}}$ if $\mathcal{E}^c$ happens.
So $V_{2}$ can be correctly decoded by (\ref{eq:case1V2}) when $\mathcal{E}$ occurs, and 
\begin{align}
	\Pr[ \hat{V}_2 \neq V_2] \leq \Pr\{\mathcal{E}\} \leq 2 \exp \lrpar{ -\tfrac{1}{8} A^2P^{1-\alpha} Q_J^2}.
\end{align}
Note that $A^2P^{1-\alpha} Q_J^2 = O(P^{\alpha-1+\epsilon })$ and $\alpha \geq 1$ in Regime 1, we have $\alpha-1+\epsilon > 0$, and $\Pr[ \hat{V}_2 \neq V_2] \rightarrow 0 $ as $P \rightarrow \infty$ as well.	
Here we conclude the proof.

\section{Proof of Theorem \ref{thm:gdofBC}} \label{sec:proofBC}

In this section, we provide the proof of Theorem \ref{thm:gdofBC}, which characterizes the SGDoF region of the ZBC with perfect and finite precision CSIT, respectively.

\subsection{The SGDoF Region with Perfect CSIT} \label{sec:proofBCPerfect}
\subsubsection{Converse}
To show the converse part, we cast the Gaussian channel model into the deterministic model defined in Section \ref{sec:det}. Lemma \ref{lemma:noloss} implies that the deterministic model incurs no loss in GDoF. 
To obtain the single-user bound for $d_1$, we apply Fano's inequality as follows.
\begin{align}
	nR_1 &\leq \Ig(\bbY_1; W_1) + \nologP \\
	&= \Ig(\bbY_1, (\bbY_1)^{ \min\{ (\beta-\alpha)^+, 1\} }; W_1) + \nologP   \label{eq:bc-d1-1} \\
	&= \Ig((\bbY_1)^{ \min\{ (\beta-\alpha)^+, 1\} }; W_1) + \Ig(\bbY_1; W_1 | (\bbY_1)^{ \min\{ (\beta-\alpha)^+, 1\} }) + \nologP  \label{eq:bc-d1-2}\\
	&\leq \Ig((\bbY_2)^{\min \{ (\beta-\alpha)^+, 1\}}; W_1) + \Hg(\bbY_1 | (\bbY_1)^{\min\{(\beta-\alpha)^+, 1\}}) + \nologP  \label{eq:bc-d1-4} \\
	&\leq n \lrpar{ \max\{\alpha, \beta \} - \min\{ (\beta-\alpha)^+, 1\}  } \log \bP + \nologP \label{eq:bc-d1-5}\\
	&= n \max\{ \alpha, \beta - 1 \} \log \bP + \nologP, \label{eq:bc-d1-6}
\end{align}
where $\bbY_1$ and $\bB$ are defined in Section \ref{sec:det}. 
Equality (\ref{eq:bc-d1-1}) holds because $(\bbY_1)^{\min\{ (\beta-\alpha)^+, 1\} }$ is a function of $\bbY_1$.
Then we apply the chain rule to get (\ref{eq:bc-d1-2}).
Next we note that, since both $(\bbY_1)^{\min\{ (\beta-\alpha)^+, 1\} }$ and $(\bbY_2)^{\min\{ (\beta-\alpha)^+, 1\} }$ contain top $\min\{ (\beta-\alpha)^+, 1\} $ segment of $\bB$ only, the latter can be obtained with the former within bounded distortion with $\chg$ given.  
Applying this observation, and by the definition of mutual information, we get inequality (\ref{eq:bc-d1-4}).
The first term in (\ref{eq:bc-d1-4}) is $\nologP$ due to Lemma \ref{lemma:noloss}, and we apply the uniform bound to obtain (\ref{eq:bc-d1-5}). 
Equality (\ref{eq:bc-d1-6}) then follows.
Finally, we arrive at $d_1 = \lim_{P \rightarrow \infty} \frac{nR_1}{n\frac{1}{2}\log P} \leq \max\{\alpha, \beta-1\}$.

Next we show the single-user bound for $d_2$ as follows. Starting by Fano's inequality, we get
\begin{align}
	nR_2 &\leq I(\bbY_2; W_2) + \nologP  \\
	&= \Ig(\bbY_2, (\bbY_2)^{\min\{1, (\beta-\alpha)^+\}};W_2) + \nologP  \label{eq:bc-d2-1}\\
	&= \Ig((\bbY_2)^{ \min\{1, (\beta-\alpha)^+\} } ; W_2) + \Ig(\bbY_2 ; W_2 | (\bbY_2)^{  \min\{1, (\beta-\alpha)^+\} })+ \nologP  \label{eq:bc-d2-2}\\
	&\leq \Ig((\bbY_1)^{ \min\{1, (\beta-\alpha)^+\} } ; W_2) + \Hg(\bbY_2 | (\bbY_2)^{ \min\{1, (\beta-\alpha)^+\} } ) + \nologP  \label{eq:bc-d2-3} \\
	&\leq n \lrpar{1 - (\beta-\alpha)^+}^+ \log \bP + \nologP, \label{eq:bc-d2-4}
\end{align}
where $\bbY_2$ is defined in (\ref{eq:modeldet2}) in Section \ref{sec:det}. 
Equality (\ref{eq:bc-d2-1}) holds because $(\bbY_2)^{ \min\{1, (\beta-\alpha)^+\} }$ is a function of $\bbY_2$.
Then we apply the chain rule to get (\ref{eq:bc-d2-2}).
Note that $(\bbY_1)^{ \min\{1, (\beta-\alpha)^+\} }$ contains the top-$\min\{1, (\beta-\alpha)^+\}$ segment of codeword $\bB$, so it can be obtained with $(\bbY_2)^{ \min\{1, (\beta-\alpha)^+\} }$ and $\chg$ within bounded distortion. 
So we apply this observation, together with the definition of mutual information, to get (\ref{eq:bc-d2-3}).
Finally we arrive at (\ref{eq:bc-d2-4}) by applying Lemma \ref{lemma:noloss} and the secrecy constraint (\ref{eq:secrecy}) to the first term in (\ref{eq:bc-d2-3}), and the uniform bound to the second term. Thus the bound $d_2 = \lim_{P \rightarrow \infty} \frac{nR_2}{n\frac{1}{2}\log P} \leq \lrpar{1-(\beta-\alpha)^+}^+$.

\subsubsection{Achievability}
To show the achievability, we present two schemes respectively for the following two regimes: (a) Regime P1: $\beta-1 \leq \alpha$, and (b) Regime P2: $\alpha < \beta-1$. 
For Regime P1, it suffices to achieve the corner point $(d_1, d_2) = (\alpha, 1-(\beta-\alpha)^+)$.
It can be achieved by zero-forcing the cross link. More specifically, we define the input codeword $X_1(t)$ and $X_2(t)$ for $t \in [n]$ as
\begin{align}
	\begin{bmatrix}
		X_1(t)\\X_2(t)
	\end{bmatrix} &= c_1(t) 
	\begin{bmatrix}
		1 \\ 0
	\end{bmatrix} U_1(t) + c_2(t) 
	\begin{bmatrix}
		-G_{12}(t) \sqrt{P^\beta} \\ G_{11}(t) \sqrt{P^\alpha}
	\end{bmatrix}U_2(t),
\end{align}
where $U_1(t)$ and $U_2(t)$ are independent codewords encoded respectively from $W_1$ and $W_2$; $c_1(t) = \frac{1}{2}$ and 
\begin{align}
	c_2(t) = \frac{1}{\sqrt{2 \lrpar{|G_{12}(t)|^2 P^\beta + |G_{11}(t)|^2 P^\alpha  }}}
\end{align}
are chosen to satisfy the unit input power constraint. Such choice of $c_2(t)$ and the precoding vector is possible because of the perfect CSIT assumption.
Note that the vector for $U_2(t)$ is chosen such that it zero-forces $U_2(t)$ at Receiver 1. Now the receivers respectively see cross-link-free channel as follows.
\begin{align}
	Y_1(t) &= \frac{1}{2} G_{11}(t) \sqrt{P^\alpha} U_1(t) + Z_1(t), \label{eq:modelzf1}\\
	Y_2(t) &= \frac{G_{22}(t) \sqrt{P^{1+\alpha}}}{ \sqrt{ 2 \lrpar{|G_{12}(t)|^2 P^\beta + |G_{11}(t)|^2 P^\alpha}  } } U_2(t) + Z_2(t).\label{eq:modelzf2}
\end{align}  
Channel (\ref{eq:modelzf1}) allows GDoF $\alpha$ for $W_1$, and channel (\ref{eq:modelzf2}) allows $1+\alpha - \max\{\alpha, \beta\} = 1-(\beta-\alpha)^+$ for $W_2$. Note that the secrecy constraint (\ref{eq:secrecy}) is satisfied, because undesired signals are zero forced and codewords $U_1(t)$ and $U_2(t)$ are independent.

On the other hand, for Regime P2, it suffices to achieve $(d_1, d_2) = (\beta -1,0)$. This can be done by setting $X_1(t) = 0$ and $X_2(t) = \sqrt{P^{-1}} U_1(t)$, where $U_1(t)$ is encoded from $W_1$ with a wiretap codebook. With such a setting, the channel allows a GDoF $\beta-1$ for $W_1$ with the secrecy constraint (\ref{eq:secrecy}) satisfied in the mean time. Here we conclude the proof.

\subsection{The SGDoF Region with Finite Precision CSIT} \label{proofBCfp}
To show $\mathcal{D}_{\tiny {BC}}^{\tiny {f.p.}}$, we continue the definition of the channel regimes in Theorem \ref{thm:gdof}, and further divide Regime 4 into the following two sub-regimes: (a) Regime 4.1, satisfying $\beta \leq 1$ and $\beta \leq \alpha$; and (b) Regime 4.2, satisfying $\beta \leq 1$ and $\alpha < \beta $. 
It remains to present the proof for Regime 4.2, as the proof for the other regimes is implied from the previous results.

More specifically, for Regime 1 and 2, their proofs follow from the proof in Section \ref{sec:wsumBound12} for the corresponding regimes, which still holds when full transmitter cooperation is allowed.
The SGDoF region of Regime 3 is identical to $\mathcal{D}_{\tiny {BC}}^{\tiny {p}}$ of the same regime, and the a achievable scheme does not rely on the perfect CSIT assumption. So the proof in Appendix \ref{sec:proofBCPerfect} holds for finite precision CSIT.
Finally, the proof for Regime 4.1 follows from the results in \cite{Chan_Geng_Jafar_secureBC}.
As a result, only the SGDoF region of Regime 4.2, which is $\{(d_1, d_2) \in \mathbb{R}_+^2:~ d_1\leq \alpha, d_1+d_2 \leq 1+\alpha-\beta\}$, remains to be shown.

First let us consider the converse proof. The single-user bound $d_1 \leq \alpha$ follows from the proof in Appendix \ref{sec:proofBCPerfect} in the corresponding channel regime. To show the sum bound, $d_1 + d_2 \leq 1+\alpha-\beta$, we cast the Gaussian ZBC model into the deterministic model defined in Section \ref{sec:det}. Lemma \ref{lemma:noloss} implies that this incurs no GDoF loss. Next we apply Fano's inequality, and get
\begin{align}
	nR_1 + nR_2 &\leq \Ig(\bbY_1; W_1) + \Ig(\bbY_2; W_2) + \nologP \label{eq:BCfp-1}\\
	&\leq \Hg(\bbY_1) - \Hg(\bbY_1 | W_1) + \Hg(\bbY_2) - \Hg(\bbY_2 | W_2) + \nologP \label{eq:BCfp-2}\\
	&= \Hg(\bbY_1|W_2) - \Hg(\bbY_1 | W_1) + \Hg(\bbY_2|W_1) - \Hg(\bbY_2 | W_2) + \nologP \label{eq:BCfp-3}\\
	&\leq \max\{1 - \beta, -\alpha  \}^+ n\log P + \max\{ \beta - 1, \alpha \}^+ n\log P \nologP\label{eq:BCfp-4}\\
	&= (1 + \alpha -\beta) n\log P +\nologP, \label{eq:BCfp-5}
\end{align}
where $\bbY_1$ and $\bbY_2$ are defined respectively in (\ref{eq:modeldet1}) and (\ref{eq:modeldet2}).
We apply (\ref{eq:noloss-2}) and the secrecy constraint (\ref{eq:secrecy}) to obtain (\ref{eq:BCfp-3}). 
Inequality (\ref{eq:BCfp-4}) holds due to Lemma \ref{lemma:ais}.
Since $\beta \leq 1$ in this regime, we have (\ref{eq:BCfp-5}), and in the GDoF limit we obtain the sum bound $d_1 + d_2 = \lim_{P \rightarrow \infty} \frac{R_1 + R_2}{\frac{1}{2}\log P} \leq 1+\alpha -\beta$.

Finally, let us consider the achievability. Since the the SGDoF region of the ZBC in Regime 4.2 is  identical to that of the ZIC in the same regime, the same achievable schemes apply. Thus, we obtain the SGDoF region of the ZBC with finite precision CSIT and conclude the proof.

\section{Proof of Lemma \ref{lemma:const}} \label{sec:const}
We assume $G_{1}$ and $G_{2}$ are real random variables with  $|G_{i}| \in (\frac{1}{\Delta}, \Delta)$ for $i = 1,2$.
For quick reference, we define $V = T \boxplus U$ and $Z = (T)^\lambda \boxplus (U)^\mu$, and summarize the definition of the top $\lambda$ sub-section of the random variables as follows:
\begin{align}
(T)^\lambda &= (T)^{\lambda+\nu}_{\nu} = \lrfloor{ \frac{ T - \bP^{\lambda+\nu} \lrfloor{\frac{T}{\bP^{\lambda+\nu}}} }{ \bP^{\nu} }} = \lrfloor{\frac{T}{\bP^{\nu}}}, \label{eq:constdef-2}\\
(U)^{\mu} &= (U)^{\mu+\nu}_{\nu} = \lrfloor{\frac{ U - \bP^{\mu+\nu} \lrfloor{\frac{U}{\bP^{\mu+\nu}} } }{ \bP^{\nu} }}  = \lrfloor{\frac{U}{\bP^{\nu}}}, \label{eq:constdef-3}\\
(V)^{\lambda} &= ( T \boxplus U)^{\lambda+\nu}_{\nu} =\lrfloor{  \frac{ V - \bP^{\lambda + \nu} \lrfloor{ \frac{V}{\bP^{\lambda+\nu} }     } }{ \bP^{\nu} }  }.  \label{eq:constdef-5}
\end{align}
Note that the last equality of (\ref{eq:constdef-2}) and (\ref{eq:constdef-3}) holds because $\lrfloor{ \frac{T}{\bP^{\lambda+\nu}}} = \lrfloor{ \frac{U}{\bP^{\mu+\nu}}} = 0$.

Next we simplify (\ref{eq:constdef-5}) in the way as is done to (\ref{eq:constdef-2}) and (\ref{eq:constdef-3}).
Define $\eta_T = G_{1}T - \lrfloor{G_1 T}$, and $\eta_U = G_{2}U - \lrfloor{G_2 U}$. Note that $\eta_T, \eta_U \in [0,1)$.
Let us first estimate the size of the support of $\lrfloor{ \frac{V}{\bP^{\lambda+\nu}} }$, which is a term appearing in the denominator of (\ref{eq:constdef-5}).
\begin{align}
\frac{V}{\bP^{\lambda+\nu}} 
&= G_{1} \frac{T}{\bP^{\lambda+\nu}} + G_{2} \frac{U}{\bP^{\lambda+\nu}} + \frac{\eta_T +\eta_U}{\bP^{\lambda+\nu}} \label{eq:const-1}\\
&= \tilde{\eta}_1 + \tilde{\eta}_2 + \tilde{\eta}_3,
\end{align}
where $\tilde{\eta}_i$ is the $i^\text{th}$ term in (\ref{eq:const-1}).
It is obvious that $\tilde{\eta}_1, \tilde{\eta}_2 \in [-\Delta, \Delta]$, and $\tilde{\eta}_3 \in [0,2]$. 
So $\lrfloor{ \frac{V}{\bP^{\lambda+\nu} } }$ is a random variable with support $\{-2\Delta, -2\Delta+1 \cdots, 0,1, \cdots, 2\Delta+2\}$.
Note that for real numbers $x, y$, we have $\lfloor x + y \rfloor =  \lfloor x \rfloor +  \lfloor y \rfloor + E$, where $E \in \{-1, 0,1\}$.  
With this observation, we can expand $(V)^{\lambda} $ defined in (\ref{eq:constdef-5}) further as follows.
\begin{align}
(V)^{\lambda} 
&=\lrfloor{\frac{V}{\bP^{\nu}} } + \underbrace{ \lrfloor{-\bP^{\lambda + \nu} \lrfloor{\frac{V}{\bP^{\lambda+\nu}}}  }  + E}_{\tilde{E}} \\
&= \lrfloor{\frac{V}{\bP^{\nu}} } + \tilde{E},
\end{align}
where $\tilde{E}$ is a random variable with support of size no greater than $3 ( 4\Delta + 3)$.

Finally we relate $Z$ to $(V)^\lambda$.
Define truncation terms $\delta_T = \frac{T}{\bP^{\nu}}  - (T)^\lambda$, $\delta_U = \frac{U}{\bP^{\nu}} - (U)^{\lambda}$, $\epsilon_T = G_{1}(T)^\lambda -\lrfloor{G_{1}(T)^\lambda}$, $\epsilon_U = G_{2}(U)^{\mu} - \lrfloor{G_{2}(U)^{\mu}}$, and $\epsilon = \frac{V}{\bP^{\nu}} - \lrfloor{ \frac{V}{\bP^{\nu}}}$, whose values are in $[0,1)$.
With these truncation terms, we relate $Z$ with $(V)^{\lambda}$ as follows.
\begin{align}
Z &= \lrfloor{ G_{1} (T)^\lambda } + \lrfloor{G_{2} (U)^\mu}\\
&= G_{1} (T)^\lambda + G_{2} (U)^\mu - (\epsilon_T + \epsilon_U)\\
&= G_{1} \frac{T}{\bP^{\nu}} + G_{2}\frac{U}{\bP^{\nu}} - (G_{1}\delta_T + G_{2} \delta_U + \epsilon_T + \epsilon_U)\\
&= \frac{1}{\bP^\nu}( \lrfloor{G_{1} T}  + \lrfloor{G_{2} U})  - \left (\frac{\eta_T+\eta_U}{\bP^{\nu}} + G_{1}\delta_T + G_{2} \delta_U + \epsilon_T + \epsilon_U \right)\\
&= \lrfloor{\frac{V}{\bP^{\nu}}}  + \epsilon -\left (\frac{\eta_T+\eta_U}{\bP^{\nu}} + G_{1}\delta_T + G_{2} \delta_U + \epsilon_T + \epsilon_U \right) \\
&= (V)^{\lambda}  - \tilde{E} - \underbrace{   \left (\frac{\eta_T+\eta_U}{\bP^{\nu}} + G_{1}\delta_T + G_{2} \delta_U + \epsilon_T + \epsilon_U -\epsilon \right)   }_{E'} \\
&= (V)^{\lambda} - \tilde{E} - E',
\end{align}
where $E'$ is a random variable taking an integer value from $[-2\Delta-1, 2\Delta+4]$ and therefore has a support of size at most $4\Delta+6$. 
As a result, $E_\Sigma = \tilde{E} + E'$ is a  random variable with a support of size at most $3(4\Delta+3)(4\Delta+6)$, which is a constant with respect to $P$.

In summary, one can evaluate $Z = (T)^\lambda \boxplus (U)^\lambda$ from $(V)^{\lambda}$ once $E_\Sigma$ is known, which is a discrete random variable with a support of constant size invariant of $P$.
By comparing the entropy of $Z$ and $(V)^{\lambda}$, we have $H(Z) - H(E_\Sigma) \leq H(Z|E_\Sigma) \leq H((V)^{\lambda}) \leq H(Z)+H(E_\Sigma)$, and therefore establish $H((T \boxplus U)^{\lambda}) = H( (T)^\lambda \boxplus (U)^\lambda )+ O(1)$.

\bibliographystyle{IEEEtran}

\end{document}